\documentclass[11pt]{article}
\usepackage[utf8]{inputenc}

\usepackage{typearea}
\typearea{15}

\usepackage{amsfonts, amsmath, amsthm, amssymb, array, caption, comment, graphicx, makecell, multirow, paralist, romanbar, subcaption, todonotes}

\usepackage{cite}

\usepackage{tikz}
\usetikzlibrary{arrows}

\usetikzlibrary{shapes,calc,arrows.meta,bending,intersections,positioning}
\tikzset{
	position/.style args={#1:#2 from #3}{
		at=(#3.#1), anchor=#1+180, shift=(#1:#2)
	}
}

\newtheorem{theorem}{Theorem}

\newtheorem{proposition}[theorem]{Proposition}
\newtheorem{lemma}[theorem]{Lemma}
\newtheorem{corollary}[theorem]{Corollary}

\newtheorem{observation}[theorem]{Observation}
\newtheorem{example}[theorem]{Example}

\newcommand{\classNP}{\textsf{NP}}

\newcommand{\tSAT}{\textsc{3-Satisfiability}}

\newcommand{\DHamilPath}{\textsc{Directed Hamiltonian Path}}

\newcommand{\vecf}{\mathbf{f}}
\newcommand{\veca}{\mathbf{a}}

\newcommand{\calM}{\mathcal{M}}
\newcommand{\TM}{threshold-${\calM}$}

\newcommand{\vectr}{\mathbf{tr}}
\newcommand{\vectp}{\mathbf{tp}}

\newcommand{\vectau}{\boldsymbol{\tau}}

\newcommand{\RR}[0]{\mathbb{R}}

\allowdisplaybreaks

\def\h{0.13 \textheight}

\newcolumntype{x}[1]{>{\centering\arraybackslash}p{#1}}
\newcommand\diag[4]{%
	\multicolumn{1}{p{#2}|}{\hskip-\tabcolsep
		$\vcenter{\begin{tikzpicture}[baseline=0,anchor=south west,inner sep=#1]
				\path[use as bounding box] (0,0) rectangle (#2+2\tabcolsep,\baselineskip);
				\node[minimum width={#2+2\tabcolsep-\pgflinewidth},
				minimum  height=\baselineskip+\extrarowheight-\pgflinewidth] (box) {};
				\draw[line cap=round] (box.north west) -- (box.south east);
				\node[anchor=south west] at (box.south west) {#3};
				\node[anchor=north east] at (box.north east) {#4};
		\end{tikzpicture}}$\hskip-\tabcolsep}}
\begin{document}
\title{Seniorities and Minimal Clearing in Financial Network Games\thanks{Supported by DFG grants Ho 3831/5-1, 6-1 and 7-1. Institute for Computer Science, Goethe University Frankfurt, Germany. \texttt{$\{$lwilhelmi,mhoefer$\}$@em.uni-frankfurt.de}}}

\author{Martin Hoefer \and Lisa Wilhelmi}
\date{}
\maketitle              % typeset the header of the contribution
\begin{abstract}
%\todo{All changes in blue}
Financial network games model payment incentives in the context of networked liabilities. %The financial network $G=(V,E,c)$ is given by a set $V$ of banks and a set $E$ of weighted, directed edges. An edge $e = (u,v) \in E$ means that bank $u$ owes a debt of $c_e > 0$ to bank $v$. Each bank is a player and uses a strategy to allocate its available assets to pay for its debt. 
In this paper, we advance the understanding of incentives in financial networks in two important directions: minimal clearing (arising, e.g., as a result of sequential execution of payments) and seniorities (i.e., priorities over debt contracts).
    
We distinguish between priorities that are chosen endogenously or exogenously. For endogenous priorities and standard (maximal) clearing, the games exhibit a coalitional form of weak acyclicity. A strong equilibrium exists and can be reached after a polynomial number of deviations. Moreover, there is a strong equilibrium that is optimal for a wide variety of social welfare functions. In contrast, for minimal clearing there are games in which no optimal strategy profile exists, even for standard utilitarian social welfare. Perhaps surprisingly, a strong equilibrium still exists and, for a wide range of strategies, can be reached after a polynomial number of deviations.  In contrast, for exogenous priorities, equilibria can be absent and equilibrium existence is \classNP-hard to decide, for both minimal and maximal clearing.
\end{abstract}
\section{Introduction}

The complex interconnections between financial institutions are a major source of systemic risk in modern economies. This has become clear over the last decade when the financial crisis unfolded. Ever since the crash of 2008, governments, regulators, and financial institutions have been making unprecedented efforts to maintain functioning and stability of the financial system (and more generally, the global economy). Perhaps surprisingly, however, little is known about the inherent challenges arising in this domain. What are sources of systemic risk in financial networks? What is the (computational) complexity of the resulting decision and optimization problems? Questions of this type are not well-understood, even in seemingly simple cases when the network is composed of elementary debt contracts. Perhaps the most fundamental operation in these networks, and the basis for most analyses of systemic risk, is \emph{clearing}, i.e., the settling of debt. Clearing is non-trivial in a network context because institutions typically depend on other institutions to satisfy their own obligations. Recent work in theoretical computer science has started to carve out interesting effects arising in this context~\cite{SchuldenzuckerSB17, SchuldenzuckerSB20, SchuldenzuckerS20, BertschingerHS20, PappW21, PappW20, PappW21wine, PappW21ec}. Despite these advances, incentives and economic implications in the clearing process are not well-understood.

In the majority of the literature, clearing is interpreted as a mechanical process where payments are prescribed to banks by a central entity following a fixed set of rules (e.g., \emph{proportional} payments w.r.t.\ amount of debt). In reality, however, financial institutions have to be expected to act strategically within the limits of their contractual and legal obligations. Therefore, in this paper, we consider financial network games (or \emph{clearing games})~\cite{BertschingerHS20}, a novel game-theoretic framework based on the classic network model by Eisenberg and Noe~\cite{EisenbergN01}. In these games, institutions strategically choose to allocate payments in order to clear as much of their debt as possible. We strive to understand the existence, structure and computational complexity of \emph{stable payments}, i.e., Nash and strong equilibria. Our focus lies on two aspects that received little or no attention in the algorithmic literature so far: seniorities and minimal clearing.

The \emph{seniority} of a debt contract is the priority that this contract enjoys in the payment order of the bank. When a bank has contracts of different seniority, it first needs to pay all debt of highest priority before spending money on obligations with lower priority. Clearing games with ``endogenous'' priorities, in which banks strategically choose the order of payment at the time of clearing, have been subject of interest very recently~\cite{BertschingerHS20,Kanellopoulos21}. Arguably, however, in financial markets such priorities are often exogenous, i.e., determined in advance and not subject to strategic choice at the time of clearing. Exogenous priorities are considered regularly across the literature~\cite{Diamond93,Fischer14,KrieterR21}, but their impact on equilibria in clearing games has not yet been considered. We are interested in the effect of limiting the strategic interaction of banks with exogenous priorities (potentially with ties), which imply strategic choices constrained by the seniority structure.

%\changed{
Clearing states represent solutions of a fixed point problem. The classic network model with proportional payments~\cite{EisenbergN01} usually guarantees a \emph{unique} clearing state -- only a few cases allow multiple clearing states to exist. In that case, the network is assumed to settle on a maximum clearing state, in which as much of the debt as possible is paid~\cite{ElsingerLS06,RogersV13}. This is justified when a central authority has structural insight into the network and the ability to steer the clearing process.

When banks choose payment functions strategically, the system allows multiple clearing states, and the choice of clearing state becomes significant. Moreover, an orchestrated clearing effort steering the process to a maximal fixed point is not immediately available, at least on a global scale, which can be seen by rather heterogeneous regulatory efforts to counter several financial crises over the last decade. Instead, clearing may occur in a decentralized, local, and sequential fashion, where institutions repay debt (partially) as soon as they have liquid assets available. In such an environment, clearing steers towards the \emph{minimal} fixed point of the payment function~\cite{CsokaH18}. Therefore, we focus on the \emph{minimal clearing state}, i.e., the set of consistent payments that clear the least debt. The properties of minimal clearing are much less understood in the literature.
%}
%In the classic network model by Eisenberg and Noe~\cite{EisenbergN01} and many of its' variants~(see, e.g., \cite{ElsingerLS06,RogersV13}), the network is assumed to settle on a set of consistent payments clearing as much of the debt as possible. Clearing occurs at the \emph{maximal fixed point} of the banks' payment function. This is justified when a central authority has structural insight into the network and the ability to steer the clearing process. In contrast, rather heterogeneous regulatory efforts to counter the financial crisis have shown that an orchestrated clearing effort is not immediately available, at least on a global scale. Instead, clearing may occur in a decentralized, local, and sequential fashion, where institutions repay debt (partially) as soon as they have liquid assets available. In such an environment, clearing steers towards the \emph{minimal} fixed point of the payment function~\cite{CsokaH18}. The computational and game-theoretic aspects of minimal clearing, however, have not yet been addressed in the literature.
%\medskip

\paragraph{Results and Contribution.}
We analyze equilibrium existence and computational complexity of pure Nash equilibria in financial clearing games with different priority and clearing properties. We distinguish between exogenous (or \emph{fixed}) and endogenous (or \emph{non-fixed}) priorities. Priorities are captured as priority orderings, possibly with \emph{thresholds} defining installments for each contract. This includes, as a special case, a priority \emph{partition} of contracts, i.e., an ordering with ties in which debt contracts are being cleared. For strategic payments to contracts with the \emph{same} priority, we consider a very general approach. We analyze both \emph{minimal} as well as standard maximal clearing. 

Formally, we define a framework of games based on a notion of so-called threshold-$\calM$ strategies. The thresholds capture the priority structure. The set $\calM$ specifies a set of functions for the payments to parts with the same priority. Most of our results apply to a very general set of threshold-$\calM$ strategies for all suitable choices of $\calM$, which is equivalent to all monotone payment strategies. Popular and natural choices for $\calM$ are strategies resulting in proportional or ranking-based payments. We will use them for exposition throughout the paper.

After defining the model in Section~\ref{sec:model}, we consider non-fixed thresholds and maximal clearing. Recent work~\cite{BertschingerHS20,Kanellopoulos21} proves non-existence and \classNP-hardness results for deciding existence of a Nash equilibrium. In contrast, we show in Section~\ref{subsec:max-clearing} that clearing games with threshold-$\calM$ strategies have a coalitional form of weak acyclicity, for every suitable set $\calM$. From every initial state, there is a sequence of coalitional deviations leading to a strong equilibrium. The sequence has length polynomial in the network size. In particular, every strategy profile with a Pareto-optimal clearing state is a strong equilibrium. This shows that the strong price of stability is 1 for many natural social welfare functions. This is a substantial extension of results~\cite{BertschingerHS20} for edge-ranking strategies with thresholds.

For non-fixed thresholds and minimal clearing, all payments have to be traced back to initial (external) assets. Our first observation in Section~\ref{subsec:min-clearing} is that these games might not admit a socially optimal strategy profile, even for the arguably most basic notion of utilitarian social welfare. Maybe surprisingly, we prove existence of strong equilibria in games with threshold-$\calM$ strategies and minimal clearing, for every suitable set $\calM$. The games are coalitional weakly acyclic if the strategies satisfy a property we call \emph{reduction consistency}. Then, from every initial state, there is a finite sequence of coalitional deviations leading to a strong equilibrium. The sequence has length polynomial in the network size. 
%Maybe surprisingly, however, games with threshold-$\calM$ strategies and minimal clearing are also coalitional weakly acyclic, for every suitable set $\calM$. Hence, they always have strong equilibria. From every initial state, there is a finite sequence of coalitional deviations leading to a strong equilibrium. The sequence has length polynomial in the network size. 
Non-existence of social optima results from a compact set of strategies and a continuity problem. If we instead restrict attention to finite strategy sets (finitely many thresholds and finite $\calM$), a social optimum trivially exists. However, we show that an optimal profile can be \classNP-hard to compute, even when a strong equilibrium can be obtained in polynomial time.

For fixed thresholds, we obtain a set of non-existence and \classNP-hardness results. They apply even when (1) some banks have fixed priorities in the form of partitions, (2) for every suitable strategy set $\calM$, and (3) for both maximal and minimal clearing. In particular, for games with fixed thresholds it can be \classNP-hard to compute an optimal strategy profile, there exist games without a Nash equilibrium, and deciding existence of a Nash equilibrium can be \classNP-hard. Notably, the results do \emph{not} follow directly from previous work in~\cite{BertschingerHS20,Kanellopoulos21} -- in previous constructions the decision about the correct \emph{priorities} is the main problem.

\paragraph{Related Work.}
% Diamond93: https://citeseerx.ist.psu.edu/viewdoc/download?doi=10.1.1.470.7389&rep=rep1&type=pdf
% Fischer14: https://papers.ssrn.com/sol3/papers.cfm?abstract_id=2367080
% KrieterR21: http://www.actuaries.org/AFIR/Colloquia/Rome2/Krieter_Rau.pdf
% (NOTE: KrieterR21 is a working paper)
%
We build upon the network model for systemic risk by Eisenberg and Noe~\cite{EisenbergN01} to study properties of clearing states for proportional payments. This model has received significant attention in the literature over the last 20 years~(see, e.g., \cite{ElsingerLS06,RogersV13,Fischer14} and references therein). Some recent additions to the literature include the effects of portfolio compression, when cyclic liabilities are removed from the network~\cite{Veraart20, SchuldenzuckerS20}.

The Eisenberg-Noe model has been extended, e.g., by augmenting the model with credit-default swaps (CDS)~\cite{SchuldenzuckerSB17, SchuldenzuckerSB20}. Computational aspectes of this model has attracted some interest recently, when banks can delete or add liabilities, donate to other banks, or change external assets~\cite{PappW20}. The model also can give rise to standard game-theoretic scenarios such as Prisoner's Dilemma. Moreover, there has been interest in debt swapping among banks as an operation to influence the clearing properties~\cite{PappW21ec}.  

We study seniorities and minimal clearing. For this model, an iterative algorithm that converges to a minimal clearing state is given in~\cite{Kusnetsov2018}. Similarly, decentralized processes that monotonically increase payments based on currently available funds converge to a minimal clearing state~\cite{CsokaH18}. Somewhat related is work on sequential defaulting in~\cite{PappW21}, where a clearing state is computed by sequentially updating the recovery rate of banks. However, this process can converge to different clearing states depending on the order of updates. Properties of the Eisenberg-Noe model with proportional payments can be extended to liabilities with different seniorities in the sense of priority orderings over debt contracts~\cite{Elsinger09}. Seniorities of this kind have been analyzed in several works, e.g., in~\cite{Diamond93,Fischer14,KrieterR21}. Notably, all works discussed so far assume a mechanical clearing process, with the vast majority focusing on proportional payments. 

Closely related to our work are recent game-theoretic approaches to clearing, where banks determine payments strategically. Financial network games without seniorities are proposed in~\cite{BertschingerHS20}, where the authors focus on ranking-based payment strategies, especially edge-ranking (priority order over debt contracts) and coin-ranking (priority order over single units of payment). Properties of clearing states in these games have also been analyzed (in a non-strategic context) in~\cite{CsokaH18}. More recent work~\cite{Kanellopoulos21} extended the results to strategies that consist of a mixture of edge-ranking and proportional payments, and to networks with CDS.

\section{Clearing Games} \label{sec:model}

\paragraph{Clearing Games.}
A clearing game is defined as follows. The financial network is modeled as a directed graph $G = (V,E)$. Each node $v \in V$ is a financial institution, which we call a \emph{bank}. Each directed edge $e = (u,v) \in E$ has an edge weight $c_e > 0$ and indicates that $u$ owes $v$ an amount of $c_e$. We denote the set of incoming and outgoing edges of $v$ by $E^-(v)=\{(u,v) \in E\}$ and $E^+(v)=\{(v,u) \in E\}$, respectively. We use $c^+_v = \sum_{e \in E^+(v)} c_e$ to denote the \emph{total liabilities}, i.e., the total amount of money owed by $v$ to other banks, and $c^-_v = \sum_{e \in E^+(v)} c_e$ to denote the total amount of money owed by other banks to $v$. Each bank $v$ has \emph{external assets} $b_v \geq 0$, which can be seen as an external inflow of money. To aid the study of computational complexity, we assume all numbers in the input, i.e., all $b_v$ and $c_e$, are integer numbers in binary encoding.

Each bank $v$ chooses a \emph{strategy} $\veca_v = (a_e)_{e \in E^+(v)}$, where $a_e : \RR \to [0,c_e]$ is a payment function for edge $e$. Given any amount $y_v \in [b_v, b_v + c^-_v]$ of total assets (external $b_v$ plus money incoming over $E^-$), bank $v$ allocates a payment $a_e(y_v)$ to edge $e$. We follow~\cite{BertschingerHS20} and assume the strategies fulfill several basic properties: (1) each $a_e(y_v)$ is monotone in $y_v$, (2) every bank spends all incoming money until all debt is paid, i.e., $\sum_{e \in E^+(v)} a_e(y_v) = \min\{y_v, c^+_v\}$, (3) no bank can generate additional money, i.e., $\sum_{e \in E^+(v)} a_e(y_v) \le y_v$.

Given a strategy profile $\veca = (\veca_v)_{v \in V}$ of the game, we assume a money flow on each edge emerges. The above conditions (1)-(3) give rise to a set of fixed-point conditions, which can be satisfied by several possible flows. More formally, a \emph{feasible flow} $\vecf = (f_e)_{e \in E}$ is such that $f_{(u,v)} = a_{(u,v)}(f_v)$ for every $(u,v) \in E$, where $f_v = b_v + \sum_{e \in E^-(v)} f_e$. The set of feasible flows for a strategy profile $\veca$ forms a complete lattice~\cite{BertschingerHS20,CsokaH18}. The \emph{clearing state} of $\veca$ is the feasible flow realized in strategy profile $\veca$. We are interested in games, where the clearing state is either the supremum (termed \emph{max-clearing}) or the infimum (\emph{min-clearing}).

Given a strategy profile $\veca$ and the resulting clearing state $\vecf$, we assume that each bank $v$ wants to maximize the total assets, i.e., the utility is $u_v(\veca,\vecf) = u_v(\vecf) = b_v + \sum_{e \in E^-(v)} f_e$. Equivalently, each bank strives to maximize its equity, given by total assets minus total liabilities $b_v + \sum_{e \in E^-(v)} f_e - c^+_v$.

Previous work has focused on priority-based strategies such as coin- or edge ranking~\cite{BertschingerHS20} or a mix of edge-ranking and proportional payments~\cite{Kanellopoulos21}. In an \emph{edge-ranking} strategy, each bank chooses a permutation $\pi_v$ over all edges in $E^+(v)$. Money is used to pay towards edges in the order of $\pi_v$. Formally, number the edges of $E^+(v)$ such that $e_1 \succ_{\pi_v} e_2 \succ_{\pi_v} \ldots \succ_{\pi_v} e_{|E^+(v)|}$. $v$ pays to $e_1$ until the payment is $c_{e_1}$, then to $e_2$ until the payment is $c_{e_2}$, then \ldots, or stops earlier if it runs out of funds. %In a \emph{coin-ranking} strategy, we assume each edge $e \in E^+(v)$ is replaced by $c_e$ many auxiliary edges of weight $1$ (each auxiliary edge represents a ``coin'' of debt). A coin-ranking strategy is simply an edge-ranking strategy over all outgoing auxiliary edges.

\paragraph{Thresholds.}
%Strategies with thresholds have been briefly touched upon in~\cite{BertschingerHS20}. A threshold $\tau_e$ for edge $e$ splits edge $e$ into a high-priority instalment $e^h$ and a low-priority one $e^l$. The weights for these parts are $c_{e^h}=\tau_e$ and $c_{e^l}= c_e - \tau_e$. Bank $v$ then first tries to pay the high-priority instalments of all edges in $E^+(v)$ according to some monotone strategy. Once these are fully paid, the remaining parts are paid according to some (potentially different) monotone strategy with the remaining funds. 
Strategies with thresholds have been briefly touched upon in~\cite{BertschingerHS20}. A threshold $\tau_e$ for edge $e$ splits edge $e$ into a high-priority installment $e^h$ and a low-priority one $e^l$. The weights for these parts are $c_{e^h}=\tau_e$ and $c_{e^l}= c_e - \tau_e$. In the following, we assume that the two parts are each represented by an auxiliary edge where $e^{(1)}$ has weight $c_{e^{(1)}}=\tau_e$ and $e^{(2)}$ has weight $c_{e^{(2)}}=c_e-\tau_e$. Bank $v$ then first tries to pay the auxiliary edge $e^{(1)}$ representing the high-priority installment, for all edges $e$ in $E^+(v)$, according to some monotone strategy. Once these are fully paid, the low-priority auxiliary edges $e^{(2)}$ are paid according to some (potentially different) monotone strategy with the remaining funds.

Towards multiple thresholds $\tau_e^{(i)}$ per edge, bank $v$ first pays auxiliary edge $e^{(1)}$ of all edges in $E^+(v)$ using a monotone strategy $\veca_v^{(1)}$. Once these are fully paid, it uses a monotone strategy $\veca_v^{(2)}$ to pay the auxiliary edges $e^{(2)}$, and so on, until all auxiliary edges are paid or $v$ runs out of funds. By assumption, at least one suitable %feasible
choice exists for $\veca^{(i)}_v$, for every possible threshold vector.

Let us formally define general classes of threshold-based strategies. A \emph{suitable} set $\calM$ of monotone strategies contains, for every vector $\vectau_v = (\tau_e)_{e \in E^+(v)}$ of thresholds with $\tau_e \in [0,c_e]$, at least one strategy $\veca \in \calM$ with $a_e(T_v) = \tau_e$ for every $e \in E^+(v)$, where $T_v = \sum_{e \in E^+(v)} \tau_e$. For any integer $k \ge 2$, a \emph{threshold-$\calM$} strategy $(\vectau_v^{(1)},\ldots,\vectau_v^{(k-1)}, \veca^{(1)}_v, \ldots, \veca^{(k)}_v)$ is given by (1) $k-1$ vectors of thresholds $\vectau^{(i)}_v = (\tau^{(i)}_e)_{e \in E^+(v)}$ with $\tau^{(i)}_e \in [0,c_e]$ and $\tau^{(i)}_e \ge \tau^{(i-1)}_e$, for all $i=1,\ldots,k-1$ and all $e \in E^+(v)$, and (2) $k$ monotone strategies $\veca^{(i)}_v \in \calM$, for $i = 1,\ldots,k$. Thresholds $\tau^{(i)}_e$ split edge $e$ into $k$ auxiliary edges $e^{(i)}$ with $c_{e^{(i)}} = \tau^{(i)}_e - \tau^{(i-1)}_e$, for $i = 1,\ldots,k$, where we assume $\tau^{(0)}_e = 0$ and $\tau^{(k)}_e = c_e$. 

%Bank $v$ first pays the parts $e^{(1)}$ of all edges in $E^+(v)$ using the monotone strategy $\veca_v^{(1)}$. Once these are fully paid, it uses monotone strategy $\veca_v^{(2)}$ to pay the parts $e^{(2)}$, and so on, until all parts are paid or $v$ runs out of funds. By assumption, $\calM$ contains at least one feasible choice for $\veca^{(i)}_v$ for all possible choices of the $k-1$ threshold vectors. 

Formally, let $T_v^{(i)} = \sum_{e \in E^+(v)} \tau_e^{(i)}$ , i.e., the sum of weights of all edges $e^{(1)}$ to $e^{(i)}$, and $t_v^{(i)} = \sum_{e \in E^+(v)} c_{e^{(i)}} = T_v^{(i)} - T_v^{(i-1)}$ , i.e., the sum of weights of all $e^{(i)}$
%where we assume $\tau^{(0)}_e = 0$ and $\tau^{(k)}_e=c_e$. 
The monotone strategy $\veca_v^{(i)} \in \calM$ satisfies $a^{(i)}_e(0) = 0$ and $a^{(i)}_e(t_v^{(i)}) = \tau_e^{(i)} - \tau_e^{(i-1)}$ $=c_{e^{(i)}}$ for all $e \in E^+(v)$. The total payment to edge $e$ in the threshold-$\calM$ strategy is then $a_e(x) = \sum_{i=1}^k a_e^{(i)}(\min\{x - T_v^{(i-1)}, t_v^{(i)}\}) = \tau_e^{(i')} + a_e^{(i'+1)}(x - T_v^{(i')})$, where $i' = \arg \max_{i < k} \{T_v^{i} \le x\}$.

\paragraph{Examples.}
%\LW{Is it more confusing to use or to drop the auxiliary-edge-interpretation in this paragraph? I'm really not sure.\\}
We discuss several natural examples for $\calM$. A \emph{threshold edge-ranking} strategy $\vectr_v = ((\vectau_v^{(i)})_{i \in [k-1]}, (\pi_v^{(j)})_{j \in [k]})$ is a threshold-$\calM$ strategy, where $\calM$ encompasses all edge-ranking strategies. Here $\pi_v^{(j)}$ are permutations over $e^{(j)}$ for all edges $e$ in $E^+(v)$. Bank $v$ pays first towards each auxiliary edge $e^{(1)}$, in the order of the ranking $\pi_v^{(1)}$, until the payment to the edge reaches $c_{e^{(1)}}$ %$\tau_e^{(1)}$ 
(or it runs out of funds). Once it pays $c_{e^{(1)}}$ %$\tau_e^{(1)}$ 
to every edge, it uses additional money to increase the payment on $e^{(2)}$ to $c_{e^{(2)}}$ %$e$ by $\tau_e^{(2)}-\tau_e^{(1)}$
, in the order of $\pi_v^{(2)}$ (or stops upon running out of funds), and so on. 

Threshold edge-ranking strategies are a significant generalization of edge-ranking strategies. In particular, we show below that clearing games with threshold edge-ranking strategies always have strong equilibria, while for edge-ranking strategies even pure Nash equilibria can be absent~\cite{BertschingerHS20}.

%We can directly extend this idea to \emph{threshold coin-ranking} strategies, where we restrict the thresholds to integers. It is easy to see that the set of threshold coin-ranking strategies (with integer thresholds) is exactly the set of coin-ranking strategies.

In a \emph{threshold proportional} strategy $\vectp_v = ((\vectau_v^{(i)})_{i \in [k-1]}, (\veca_v^{(j)})_{j \in [k]})$, the set $\calM$ encompasses all proportional payment strategies. Bank $v$ pays the marginal debt on auxiliary edges $e^{(i)}$ of class $i$ using proportional payments in this class:
\[
    a_e^{(j)}(x) = x \cdot \frac{c_{e^{(j)}}}{\sum_{e \in E^+(v)} c_{e^{(j)}}} = x \cdot \frac{\tau_e^{(j)} - \tau_e^{(j-1)}}{t_v^{(j)}}\enspace,
\]
in the order of $j=1,\ldots,k$. Fixing any set of thresholds $(\vectau_v^{(i)})_{i \in [k-1]}$, proportional strategies $\veca_v^{(j)} \in \calM$ are uniquely determined for every $j=1,\ldots,k$ (while for edge-ranking strategies there are $(E^+(v))!$ many choices for each $\veca_v^{(j)}$).

\paragraph{Thresholds and Partitions.}
An interesting special case of a threshold-$\calM$ strategy is a \emph{partition-$\calM$} strategy, in which all thresholds are $\tau_e^{(i)} \in \{0,c_e\}$. Then each auxiliary edge 
%\MH{First time we mention auxiliary edges... remove? Rephrase? Define? Also used in the upcoming observation...} \LW{They are now first mentioned and explained in the threshold-paragraph.} 
has weight $c_{e^{(i)}} \in \{0,c_e\}$. The result is a partition ranking over the contracts, where for partition $E^{(i)}$ of priority $i$ we make payments based on the strategy $a_v^{(i)}$. Conversely, every threshold-${\calM}$ strategy $(\vectau_v^{(1)},\ldots,\vectau_v^{(k-1)},$ $\veca^{(1)}_v, \ldots, \veca^{(k)}_v)$ can be interpreted as a partition-$\calM$ strategy over auxiliary edges $e^{(i)}$ for all $e \in E^+(v)$.
\begin{observation}
    \label{obs:thresh-to-part}
    Every threshold-${\calM}$ strategy can be interpreted as a partition-$\calM$ strategy where each edge $e \in E^+{(v)}$ is replaced by $k$ auxiliary edges.
\end{observation}
\begin{figure}[ht!]
	\begin{subfigure}[t]{\textwidth}
    \centering
    \resizebox{0.14 \textheight}{!}{
	\begin{tikzpicture}
		\def\r{1}
		%nodes
		\node[circle, draw=black, inner sep=2.5pt] (1) at (210:\r) {$v_1$};
		\node[circle, draw=black, inner sep=2.5pt] (2) at (330:\r) {$v_2$};
		\node[circle, draw=black, inner sep=2.5pt] (3) at (90:\r) {$v_3$};
		\node[circle, draw=black, inner sep=2.5pt] (4) [below = of 1] {$v_4$};
		\node[circle, draw=black, inner sep=2.5pt] (5) [below = of 2] {$v_5$};
		%external assets
		\node[rectangle, draw=black, inner sep=2.5pt] (bv) [left=0.1cm of 1] {$1$};
		%edges
		\draw[->] (1) -- node[above, midway] {\small $1$} (2);
		\draw[->] (1) -- node[left, midway] {\small $1$} (4);
		\draw[->] (2) -- node[right, midway] {\small $1$} (3);
		\draw[->] (2) -- node[right, midway] {\small $1$} (5);
		\draw[->] (3) -- node[left, midway] {\small $1$} (1);
	\end{tikzpicture}
	}
	\caption{Game of Example \ref{ex:intro}. $v_1$ has external assets of 1, all edges have a weight of 1.}
	\label{fig:intro-example-game}
    \end{subfigure}
    
    \vspace{0.1cm} 
    
	%\hspace{-1cm}
	%%%% max clearing, strategy profile 1 %%%%
	\begin{subfigure}[t]{\textwidth}
	\centering
	\resizebox{\h}{!}{
	\begin{tikzpicture}
		\def\r{1}
		%nodes
		\node[circle, draw=black, inner sep=2.5pt] (1) at (210:\r) {$v_1$};
		\node[circle, draw=black, inner sep=2.5pt] (2) at (330:\r) {$v_2$};
		\node[circle, draw=black, inner sep=2.5pt] (3) at (90:\r) {$v_3$};
		\node[circle, draw=black, inner sep=2.5pt] (4) [below = of 1] {$v_4$};
		\node[circle, draw=black, inner sep=2.5pt] (5) [below = of 2] {$v_5$};
		%external assets
		\node[rectangle, draw=black, inner sep=2.5pt] (bv) [left=0.1cm of 1] {$1$};
		%edges
		\draw[-stealth',line width=1.5pt] (1) -- node[above, midway] {\small $1$} (2);
		\draw[->] (1) -- node[left, midway] {\small $1$} (4);
		\draw[-stealth',line width=1.5pt] (2) -- node[right, midway] {\small $1$} (3);
		\draw[->, gray] (2) -- (5);
		\draw[->] (3) -- node[left, midway] {\small $1$} (1);
	\end{tikzpicture}
	}
	\hspace{0.25cm}
	\resizebox{\h}{!}{
	\begin{tikzpicture}
		\def\r{1}
		%nodes
		\node[circle, draw=black, inner sep=2.5pt] (1) at (210:\r) {$v_1$};
		\node[circle, draw=black, inner sep=2.5pt] (2) at (330:\r) {$v_2$};
		\node[circle, draw=black, inner sep=2.5pt] (3) at (90:\r) {$v_3$};
		\node[circle, draw=black, inner sep=2.5pt] (4) [below = of 1] {$v_4$};
		\node[circle, draw=black, inner sep=2.5pt] (5) [below = of 2] {$v_5$};
		%external assets
		\node[rectangle, draw=black, inner sep=2.5pt] (bv) [left=0.1cm of 1] {$1$};
		%edges
		\draw[->] (1) -- node[above, midway] {\small $1$} (2);
		\draw[-stealth',line width=1.5pt] (1) -- node[left, midway] {\small $1$} (4);
		\draw[-stealth',line width=1.5pt] (2) -- node[right, midway] {\small $1$} (3);
		\draw[->, gray] (2) -- (5);
		\draw[->] (3) -- node[left, midway] {\small $1$} (1);
	\end{tikzpicture}
	}
	\hspace{0.25cm}
	%
	%%%%max clearing, strategy profile 2 %%%%
	\resizebox{\h}{!}{
	\begin{tikzpicture}
		\def\r{1}
		%nodes
		\node[circle, draw=black, inner sep=2.5pt] (1) at (210:\r) {$v_1$};
		\node[circle, draw=black, inner sep=2.5pt] (2) at (330:\r) {$v_2$};
		\node[circle, draw=black, inner sep=2.5pt] (3) at (90:\r) {$v_3$};
		\node[circle, draw=black, inner sep=2.5pt] (4) [below = of 1] {$v_4$};
		\node[circle, draw=black, inner sep=2.5pt] (5) [below = of 2] {$v_5$};
		%external assets
		\node[rectangle, draw=black, inner sep=2.5pt] (bv) [left=0.1cm of 1] {$1$};
		%edges
		\draw[-stealth',line width=1.5pt] (1) -- node[above, midway] {\small $1$} (2);
		\draw[->, gray] (1) -- (4);
		\draw[->, gray] (2) -- (3);
		\draw[-stealth',line width=1.5pt] (2) -- node[right, midway] {\small $1$} (5);
		\draw[->, gray] (3) -- (1);
	\end{tikzpicture}
	}
	\hspace{0.25cm}
	\resizebox{\h}{!}{
	\begin{tikzpicture}
		\def\r{1}
		%nodes
		\node[circle, draw=black, inner sep=2.5pt] (1) at (210:\r) {$v_1$};
		\node[circle, draw=black, inner sep=2.5pt] (2) at (330:\r) {$v_2$};
		\node[circle, draw=black, inner sep=2.5pt] (3) at (90:\r) {$v_3$};
		\node[circle, draw=black, inner sep=2.5pt] (4) [below = of 1] {$v_4$};
		\node[circle, draw=black, inner sep=2.5pt] (5) [below = of 2] {$v_5$};
		%external assets
		\node[rectangle, draw=black, inner sep=2.5pt] (bv) [left=0.1cm of 1] {$1$};
		%edges
		\draw[->, gray] (1) -- (2);
		\draw[-stealth',line width=1.5pt] (1) -- node[left, midway] {\small $1$} (4);
		\draw[->, gray] (2) -- (3);
		\draw[-stealth',line width=1.5pt] (2) -- (5);
		\draw[->, gray] (3) -- (1);
	\end{tikzpicture}
	}
	\caption{Maximal clearing states}
	\label{fig:intro-example-max}
	\end{subfigure}
	
    \vspace{0.1cm} 

	%%%% min clearing, strategy profile 1 %%%%
%	\vspace{-1.5cm}\\ 
%	\begin{subfigure}{0.4\textwidth}
%		\hspace{0.4\textwidth}
%	\end{subfigure}
%	\hfill
	\begin{subfigure}[t]{\textwidth}
	\centering
	\resizebox{\h}{!}{
	\begin{tikzpicture}
		\def\r{1}
		%nodes
		\node[circle, draw=black, inner sep=2.5pt] (1) at (210:\r) {$v_1$};
		\node[circle, draw=black, inner sep=2.5pt] (2) at (330:\r) {$v_2$};
		\node[circle, draw=black, inner sep=2.5pt] (3) at (90:\r) {$v_3$};
		\node[circle, draw=black, inner sep=2.5pt] (4) [below = of 1] {$v_4$};
		\node[circle, draw=black, inner sep=2.5pt] (5) [below = of 2] {$v_5$};
		%external assets
		\node[rectangle, draw=black, inner sep=2.5pt] (bv) [left=0.1cm of 1] {$1$};
		%edges
		\draw[-stealth',line width=1.5pt] (1) -- node[above, midway] {\small $1$} (2);
		\draw[->] (1) -- node[left, midway] {\small $1$} (4);
		\draw[-stealth',line width=1.5pt] (2) -- node[right, midway] {\small $1$} (3);
		\draw[->, gray] (2) -- (5);
		\draw[->] (3) -- node[left, midway] {\small $1$} (1);
	\end{tikzpicture}
	}
	\hspace{0.25cm}
	\resizebox{\h}{!}{
	\begin{tikzpicture}
		\def\r{1}
		%nodes
		\node[circle, draw=black, inner sep=2.5pt] (1) at (210:\r) {$v_1$};
		\node[circle, draw=black, inner sep=2.5pt] (2) at (330:\r) {$v_2$};
		\node[circle, draw=black, inner sep=2.5pt] (3) at (90:\r) {$v_3$};
		\node[circle, draw=black, inner sep=2.5pt] (4) [below = of 1] {$v_4$};
		\node[circle, draw=black, inner sep=2.5pt] (5) [below = of 2] {$v_5$};
		%external assets
		\node[rectangle, draw=black, inner sep=2.5pt] (bv) [left=0.1cm of 1] {$1$};
		%edges
		\draw[->, gray] (1) -- (2);
		\draw[-stealth',line width=1.5pt] (1) -- node[left, midway] {\small $1$}(4);
		\draw[-stealth',line width=1.5pt] (2) -- (3);
		\draw[->, gray] (2) -- (5);
		\draw[->, gray] (3) -- (1);
	\end{tikzpicture}
	}
	\hspace{0.25cm}
    %
	%%%% min clearing, strategy profile 2 %%%%%
	\resizebox{\h}{!}{
	\begin{tikzpicture}
		\def\r{1}
		%nodes
		\node[circle, draw=black, inner sep=2.5pt] (1) at (210:\r) {$v_1$};
		\node[circle, draw=black, inner sep=2.5pt] (2) at (330:\r) {$v_2$};
		\node[circle, draw=black, inner sep=2.5pt] (3) at (90:\r) {$v_3$};
		\node[circle, draw=black, inner sep=2.5pt] (4) [below = of 1] {$v_4$};
		\node[circle, draw=black, inner sep=2.5pt] (5) [below = of 2] {$v_5$};
		%external assets
		\node[rectangle, draw=black, inner sep=2.5pt] (bv) [left=0.1cm of 1] {$1$};
		%edges
		\draw[-stealth',line width=1.5pt] (1) -- node[above, midway] {\small $1$} (2);
		\draw[->, gray] (1) -- (4);
		\draw[->, gray] (2) -- (3);
		\draw[-stealth',line width=1.5pt] (2) -- node[right, midway] {\small $1$} (5);
		\draw[->, gray] (3) -- (1);
	\end{tikzpicture}
	}
	\hspace{0.25cm}
	\resizebox{\h}{!}{
	\begin{tikzpicture}
		\def\r{1}
		%nodes
		\node[circle, draw=black, inner sep=2.5pt] (1) at (210:\r) {$v_1$};
		\node[circle, draw=black, inner sep=2.5pt] (2) at (330:\r) {$v_2$};
		\node[circle, draw=black, inner sep=2.5pt] (3) at (90:\r) {$v_3$};
		\node[circle, draw=black, inner sep=2.5pt] (4) [below = of 1] {$v_4$};
		\node[circle, draw=black, inner sep=2.5pt] (5) [below = of 2] {$v_5$};
		%external assets
		\node[rectangle, draw=black, inner sep=2.5pt] (bv) [left=0.1cm of 1] {$1$};
		%edges
		\draw[->, gray] (1) -- (2);
		\draw[-stealth',line width=1.5pt] (1) -- node[left, midway] {\small $1$} (4);
		\draw[->, gray] (2) -- (3);
		\draw[-stealth',line width=1.5pt] (2) -- (5);
		\draw[->, gray] (3) -- (1);
	\end{tikzpicture}
	}
	\caption{Minimal clearing states}
	\label{fig:intro-example-min}
	\end{subfigure}
	\caption{Image (a) depicts a clearing game whereas images in (b) and (c) represent the maximal and minimal clearing states resulting from all four edge-ranking strategy profiles. Edge labels denote flows over the edge, there is no flow on edges without labels. Thick edges indicate the preferred outgoing edge in the strategy of the node. \label{fig:intro-example}}
\end{figure}
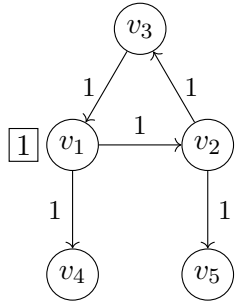
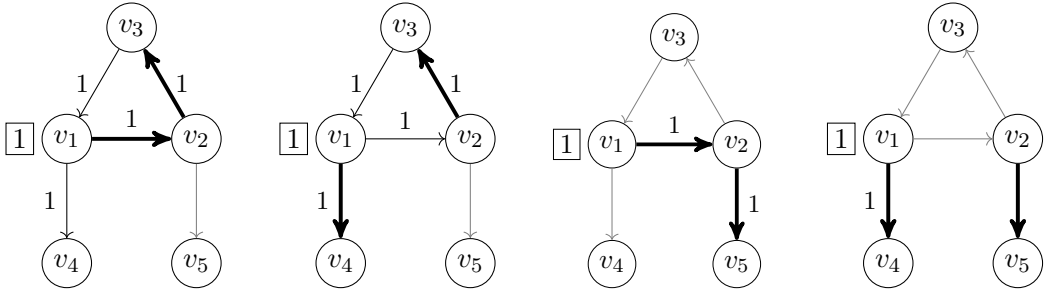
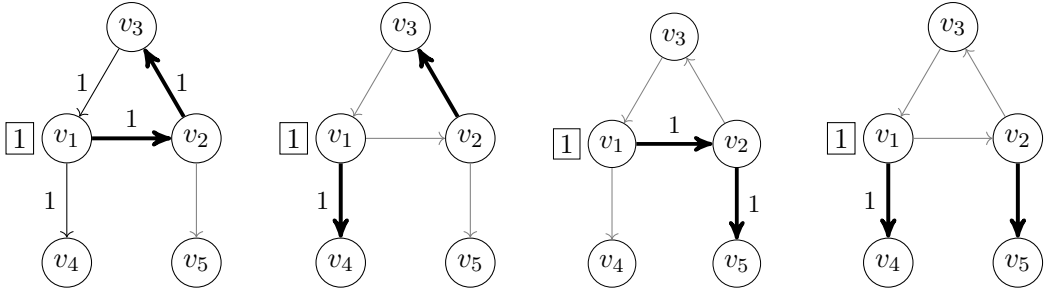

Let us illustrate the relationship between strategy profiles and the associated clearing states. For simplicity, we assume edge-ranking strategy profiles (where all thresholds $\tau_e = c_e$). The following example builds an intuition for the differences between minimal and maximal clearing states of a given strategy profile. Further, it illustrates the great benefit that the players derive from collaboration.

\begin{example}\label{ex:intro}\rm
Consider the game depicted in Fig.\ \ref{fig:intro-example-game}. Note that the behavior of players $v_3, v_4$ and $v_5$ is fully predetermined, since they either have no liabilities or pay all assets towards their single outgoing edge. Hence, it is sufficient to focus on the (non-trivial) strategy choices of players $v_1$ and $v_2$. In Fig.~\ref{fig:intro-example-max} and~\ref{fig:intro-example-min} we show the resulting maximal and minimal clearing states for all four strategy profiles, depending on the edge rankings chosen by $v_1$ and $v_2$. In the following, we discuss the clearing states for two profiles. 

Consider strategy profile $\veca$ (depicted as the second from the left in Fig.\ \ref{fig:intro-example}), where player $v_1$ first fully pays off edge $(v_1,v_4)$ and afterwards $(v_1, v_2)$, whereas player $v_2$ completely pays off edge $(v_2, v_3)$ before paying edge $(v_2, v_5)$. Formally, the strategies are $\pi^{(1)}_{v_1} = ((v_1,v_4), (v_1,v_2))$ and $\pi^{(1)}_{v_2} = ((v_2,v_3), (v_2,v_5))$. Each player $v_i$, for $i \in \{1,2\}$, directs all payments towards her highest-ranked edge until all debt is settled. All remaining money is used to pay off the other edge.  
%Formally, these strategies are defined by the thresholds $\tau_{(v_1, v_4)}^{(1)} = \tau_{(v_1, v_4)}^{(2)} = 1$ and $\tau_{(v_1, v_2)}^{(1)} = 0, \tau_{(v_1, v_2)}^{(2)} = 1$ for player $v_1$. Similarly, the thresholds for the strategy of $v_2$ are given by $\tau_{(v_2, v_3)}^{(1)} = \tau_{(v_2, v_3)}^{(2)} = 1$ and $\tau_{(v_2, v_5)}^{(1)} = 0$, $\tau_{(v_2, v_5)}^{(2)} = 1$. Here, the choices of strategies $\veca_v^{(j)}$ are inconsequential, since the players pay off edges completely one after the other. 
Let us first consider the resulting minimal clearing state. Player $v_1$ holds external assets of 1 and is obligated to spend these. By her strategy choice, she just pays off her debt of 1 towards $v_4$. The resulting payments satisfy all fixed-point conditions and represent the minimal clearing state for $\veca$. To construct the maximal clearing state, we examine if there are higher feasible payments. Note that players $v_1, v_2$ and $v_3$ with edges $(v_1,v_2), (v_2,v_3)$ and $(v_3,v_1)$ form a cycle in terms of the (most preferred) payments. Suppose the payment along the cycle is increased by 1. The fixed-point conditions remain satisfied. It is easy to see that no larger feasible payments can exist. The maximal clearing state is attained.  

Now consider the second strategy profile $\veca'$ (depicted as the third from the left in Fig.\ \ref{fig:intro-example}), where the order of payments is reversed for both players. Here the strategies are $\pi^{(1)}_{v_1} = ((v_1,v_2), (v_1,v_4))$ and $\pi^{(1)}_{v_2} = ((v_2,v_5), (v_2,v_3))$.
%Hence, this behavior can be formalized by the thresholds $\tau_{(v_1, v_2)}^{(1)} = \tau_{(v_1, v_2)}^{(2)} = 1$ and $\tau_{(v_1, v_4)}^{(1)} = 0, \tau_{(v_1, v_4)}^{(2)} = 1$ for player $v_1$ and $\tau_{(v_2, v_5)}^{(1)} = \tau_{(v_2, v_5)}^{(2)} = 1$ and $\tau_{(v_2, v_3)}^{(1)} = 0$, $\tau_{(v_2, v_3)}^{(2)} = 1$ for $v_2$. 
In this case, $v_1$ pays off her liabilities towards $v_2$ using her external assets, whereupon $v_2$ passes the incoming payments on to the sink $v_5$. Again, the resulting payments are feasible and form the minimal clearing state. In contrast to $\veca$, for $\veca'$ the clearing state is unique, i.e., the minimal clearing state equals the maximal clearing state. To observe this, consider the cycle including $v_1, v_2$ and $v_3$. The edge $(v_1,v_2)$ is already fully saturated and, hence, payments along the cycle cannot by increased. Instead, player $v_1$ would pass all additional incoming payments to $v_4$, which excludes the possibility of a feasible increase of payments.
\hfill $\blacksquare$
\end{example}

\section{Clearing Games with Seniorities}

\subsection{Max-Clearing} \label{subsec:max-clearing}

In this section, we consider computational problems in clearing games with max-clearing. Subclasses of our games have been studied in the literature before. In particular, for edge-ranking strategies, non-existence of Nash equilibrium and \classNP-hardness of deciding equilibrium existence follow directly from~\cite{BertschingerHS20}. Similar results can be shown for proportional strategies when thresholds are $\tau_e^{(i)} \in \{0,c_e\}$, i.e., for \emph{partition} proportional strategies considered in~\cite{Kanellopoulos21}. We instead focus on threshold variants and start with an observation about the clearing state.
\begin{proposition}
    \label{prop:maxClear}
    A maximum clearing state can be computed in polynomial time for games with threshold edge-ranking or threshold proportional strategies.
\end{proposition}
\begin{proof}
    We use Observation~\ref{obs:thresh-to-part}. For any threshold edge-ranking strategy, Observation~\ref{obs:thresh-to-part} leads to an equivalent partition edge-ranking strategy, i.e., a priority partition of auxiliary edges. Within each partition, the edge-ranking strategy refines the ranking to a complete ranking over all auxiliary edges. Hence, a threshold edge-ranking strategy is a regular edge-ranking strategy over auxiliary edges. This allows to compute a maximal clearing state using the algorithm in~\cite{BertschingerHS20}. 
    
    A similar observation holds for threshold proportional strategies. A threshold proportional strategy is a partition proportional strategy over auxiliary edges. Hence, the algorithm in~\cite{Kanellopoulos21} can be used to compute a maximal clearing state. 
\end{proof}

Given a strategy profile $\veca$ and the corresponding clearing state $\vecf$, consider a set $C \subseteq V$ of players and a deviation $\veca'_C = (\veca'_v)_{v \in C}$. Let $\veca' = (\veca'_C, \veca_{-C})$ be the resulting strategy profile after deviation, and let $\vecf'$ be the clearing state. The pair $(\veca, \veca')$ is a \emph{coalitional improvement step} if $u_v(\veca', \vecf') > u_v(\veca,\vecf)$ for all $v\in C$, i.e., every player strictly improves her utility upon deviation.
Furthermore, a strategy profile $\veca$ with corresponding clearing state $\vecf$ forms a \emph{strong equilibrium}, if there exists \emph{no coalitional improvement step} $(\veca, \veca')$, for any coalition $C \subseteq V$.

For the main result in this section, we show that for every game with threshold-${\calM}$ strategies and every initial strategy profile, there is a sequence of coalitional improvement steps ending in a strong equilibrium. We term this property \emph{coalitional weakly acyclic}.

\begin{theorem}
    \label{thm:max-clearing-SE}
    Every max-clearing game with threshold-${\calM}$ strategies is coalitional weakly acyclic. For every initial strategy profile, there is a sequence of coalitional improvement steps that ends in a strong equilibrium. The sequence requires a number of steps that is polynomial in $|E|$.
\end{theorem}

We show the theorem in two steps. Consider a strategy profile $\veca$ of threshold-$\calM$ strategies with \emph{Pareto-optimal clearing state} $\vecf$, i.e., there is no other strategy profile $\veca'$ with clearing state $\vecf'$ such that $f'_e \ge f_e$ for all $e \in E$ and strict inequality for at least one $e \in E$.

\begin{lemma}
    \label{lem:max-clearing-SE-lem1}
    Every strategy profile in a max-clearing game with Pareto-optimal clearing state is a strong equilibrium.
\end{lemma}
\begin{proof}
    For any game with threshold-${\calM}$ strategies and any initial state, consider the resulting money flow $\vecf$ in the maximal clearing state. Now consider any improvement step of coalition $C \subseteq V$ and a player $v \in C$. Since $u_v((\veca'_C,\veca_{-C}), \vecf') > u_v(\veca,\vecf)$, player $v$ has strictly larger total assets after the deviation. As such, there must be one incoming edge $e \in E^{-}(v)$ for which the money flow on that edge has strictly increased to $f'_{e} > f_e$. Let $w$ be the source of $e$. If $w \in C$, we can repeat the argument for $w$. If $w \not\in C$, then strategy $\veca_w$ remains the same before and after the deviation. $\veca_w$ is monotone, so a strictly larger flow $f'_e > f_e$ means that the total assets of $w$ have strictly increased $u_w((\veca'_C,\veca_{-C}), \vecf') > u_w(\veca,\vecf)$. Thus, if $w \not\in C$ we can repeat the argument for $w$.
    
    Therefore, if there is any coalitional improvement step, then we can start from any node $v \in C$ and move in reverse direction along incoming edges with strictly increased flow. This must eventually lead into a cycle $K$, for which the flow has strictly increased on every edge of $E(K)$. In turn, if the clearing state $\vecf$ is Pareto-optimal, it does not allow a cycle along which flow can be increased. Therefore, we must be in a strong equilibrium. 
\end{proof}

\begin{lemma}
    \label{lem:max-clearing-SE-lem2}
    Consider any profile $\veca$ with clearing state $\vecf$ which is not Pareto-optimal. There is a coalitional improvement step $(\veca, \veca')$ such that (1) the clearing state $\vecf'$ in $\veca'$ Pareto-dominates $\vecf$ and (2) $\vecf'$ saturates at least one more edge than $\vecf$.
\end{lemma}
\begin{proof}
    Given strategy profile $\veca$, consider clearing state $\vecf$. By assumption, there is another profile with clearing state $\tilde{\vecf}$ that Pareto dominates $\vecf$. It is easy to see~\cite[Proposition 2]{BertschingerHS20} that all feasible flows in clearing games can be interpreted as circulation flows. As a consequence, the difference $\delta_e = \tilde{f}_e - f_e \ge 0$ is again a circulation flow and, due to Pareto domination, contains at least one cycle $K$ with a positive flow. Hence, $\vecf$ admits a cycle $K$ such that $f_e < c_e$ for all $e \in E(K)$. 
    
    We construct a profitable coalitional deviation using an intermediate step. First, all banks $v \in V(K)$ deviate to any threshold-$\calM$ strategy with $\tau_e^{(1)} = f_e$ for all $e \in E^+(v)$ and $\veca_v^{(1)}, \veca_v^{(2)} \in \calM$ chosen arbitrarily but in correspondence with the new thresholds $\vectau^{(1)}$. The banks in $V(K)$ simply set all their thresholds to the flows in $\vecf$. Note that $\vecf$ remains a feasible flow in this intermediate profile. Now let $\delta = \min_{e \in E(K)} c_e - f_e > 0$ be the minimum residual weight on cycle $K$. We let each player $v \in V(K)$ deviate to a threshold-$\calM$ strategy $\veca'_v$ with ${\tau'}_e^{(1)} = f_e$ for all $e \in E^+(v) \setminus E(K)$, ${\tau'}_e^{(1)} = f_e + \delta$ for the unique $e = (v,w) \in E(K)$, and ${\veca'}_v^{(1)}, {\veca'}_v^{(2)} \in \calM$ chosen arbitrarily but in correspondence with the new thresholds ${\vectau'}^{(1)}$. After this deviation, there is a feasible flow $\vecf'$ with $f'_e = f_e$ for all $e \in E \setminus E(K)$ and $f'_e = f_e + \delta$ for all $e \in E(K)$. Moreover, this feasible flow strictly improves the assets of all banks $v \in V(K)$ by $\delta > 0$. Since the maximal clearing state Pareto-dominates any feasible flow in terms of edge flows and assets, we have constructed a coalitional improvement step for coalition $C = V(K)$. The clearing state $\vecf'$ in $\veca'$ Pareto-dominates $\vecf$ and saturates at least one more edge on cycle $K$.
\end{proof}

Since the coalitional improvement step in Lemma~\ref{lem:max-clearing-SE-lem2} saturates one more edge with flow, any sequence of such steps has length at most $|E|$. This proves the theorem.

Consider any objective function $z(\vecf)$ that is \emph{flow-monotone}, i.e., if $\vecf'$ Pareto-dominates $\vecf$, then $z(\vecf') > z(\vecf)$. A wide variety of natural social welfare functions in clearing games enjoy this property, e.g., the utilitarian social welfare $z(\vecf) = \sum_{v \in V}(b_v + \sum_{e \in E^-(v)}f_e)$ or the total debt that is repaid $z(\vecf) = \sum_{v \in V} u_v(\vecf) = \sum_{e \in E} f_e$, the total number of liquid banks $z(\vecf) = |\{ v \in V \mid b_v + \sum_{e \in E^-(v)} f_e \ge c^+ \}|$, the Nash social welfare $z(\vecf) = \left(\prod_{v \in V} u_v(\vecf)\right)^{1/n}$, and many more.

Given a flow-monotone welfare function $z$, consider a \emph{$z$-optimal money flow} that satisfies weak flow conservation, i.e., a flow $\vecf^*$ that maximizes $z(\vecf^*)$ subject to $\sum_{e \in E^+(v)} f_e^* \le \sum_{e \in E^-(v)} f_e^* + b_v$ for every node $v \in V$ and $f_e^* \le c_e$ for every edge $e \in E$. Clearly, $\vecf^*$ is Pareto-optimal among all feasible flows for all strategy profiles. Inspecting the proof of Lemma~\ref{lem:max-clearing-SE-lem2}, $\vecf^*$ can be realized as a clearing state of a strategy profile, where every player $v$ just uses a $2$-threshold-$\calM$ strategy with $\tau_e = f_e^*$ for all $e \in E^+(v)$. Clearly, this strategy profile is a strong equilibrium due to Lemma~\ref{lem:max-clearing-SE-lem1}. Note that here further structural properties (such as edge-ranking or proportionality) of strategies $a_v^{(1)}, a_v^{(2)} \in \calM$ for all $v \in V$ are inconsequential -- the maximal clearing state yields a flow $f_e = \tau_e^{(1)}$ for all $e \in E$.

Formally, this allows to bound the \emph{strong price of stability}, which denotes the ratio of the social welfare in the social optimum and the highest social welfare of all strong equilibria.

\begin{corollary}
    \label{cor:maxClearPoS}
     In every max-clearing game with threshold-${\calM}$ strategies, the strong price of stability is 1 for every flow-monotone welfare function. There is an optimal strong equilibrium with the same clearing state for every choice of $\calM$.
\end{corollary}

\subsection{Min-Clearing}\label{subsec:min-clearing}
Let us turn to games with min-clearing. In a minimal clearing state, all flow is initiated by external assets. In particular, consider an iterative process where players initially pay off debt only utilizing their external assets. This creates an initial flow. Then, in the next round, players may use additional incoming assets to pay off further debt, and so on. For a given strategy profile, the iteration of this process will monotonically increase the flow towards a minimal clearing state (see also~\cite{CsokaH18} for more formal arguments of this fact). Indeed, this idea can be applied in a structured fashion for threshold edge-ranking games to compute a minimal clearing state in polynomial time. 

\begin{proposition}
    \label{prop:minClear}
    A minimal clearing state can be computed in polynomial time for games with threshold edge-ranking strategies.
\end{proposition}

\begin{proof}
    For Proposition~\ref{prop:maxClear} we saw that we can interpret threshold edge-ranking strategies as edge-ranking strategies over auxiliary edges. Hence, we can apply the algorithm in~\cite{BertschingerHS20}, where we stop the algorithm after resolving all ``necessary cycles'' (but none of the ``optional cycles''). This results in a minimal clearing state.
\end{proof}

In contrast, observations in \cite[Chapter 3]{Kusnetsov2018} suggest that the minimal clearing state for proportional payments can not be computed with a bounded number of steps, even without thresholds (or $\tau_e^{(i)} = c_e$ for all $i \le k$). Furthermore, the non-existence of pure Nash equilibria for partition edge-ranking and proportional strategies shown in \cite[Proposition 13]{BertschingerHS20} and \cite[Theorem 4]{Kanellopoulos21} transfers to min-clearing -- it can be verified that there are \emph{unique} feasible flows for any strategy profile arising in the games without pure equilibria. 

We first observe that Corollary~\ref{cor:maxClearPoS} cannot be extended to min-clearing games. The main problem is that there might be no social optimum, even for utilitarian social welfare $z(\vecf) = \sum_{v \in V} u_v(\veca,\vecf)$. 

\begin{example}\label{ex:noOpt} \rm
Consider threshold-$\calM$ strategies in the game depicted in Figure~\ref{fig:min-clearing-undefined-opt}, for any suitable set $\calM$. Player $v$ is the only player with a non-trivial strategy choice. Also $v$ is the only player with external assets and, hence, all money flow must be originated by $v$. We, therefore, want to find a strategy for $v$ that maximizes social welfare. If $v$ pays all external assets either towards $u_1$ or $w_1$, the social welfare is $1 + n$ in both cases. Now assume $v$ splits assets and picks a strategy with threshold $\epsilon$ on $(v,w_1)$ and $1-\epsilon$ on $(v,u_1)$. Then payments correspond to thresholds. The payment of $\epsilon$ is repeatedly passed through the cycle until all players $w_i$ have completely paid off their liabilities. The payment of $1-\epsilon$ is passed on along the path to $u_n$. The social welfare $1 + n + \epsilon + n\cdot(1-\epsilon) = 1 + 2n - \epsilon\cdot(n-1)$ is monotonically increasing when $\epsilon \to 0$. However, at $\epsilon=0$ no flow is initiated on the cycle, and the social welfare drops to $1+n$. Consequently, there is no social optimum for this instance. \hfill $\blacksquare$
\end{example}

\begin{figure}[t]
    \centering
    \resizebox{!}{0.16 \textwidth}{
    \begin{tikzpicture}[
	pics/circular arc/.style args={from #1 to #2}{code={
			\path[name path=arc] 
			let \p1=(#1),\p2=(#2),\n1={Mod(720+atan2(\y1,\x1),360)},
			\n2={Mod(720+atan2(\y2,\x2),360)},
			\n3={ifthenelse(abs(\n1-\n2)<180,\n2,\n2+360)}
			in (\n1:\r) arc(\n1:\n3:\r);
			\draw[pic actions,
			name intersections={of=#1 and arc,by=arcstart},
			name intersections={of=#2 and arc,by=arcend}] 
			let \p1=(arcstart),\p2=(arcend),\n1={Mod(720+atan2(\y1,\x1),360)},
			\n2={Mod(720+atan2(\y2,\x2),360)},
			\n3={ifthenelse(abs(\n1-\n2)<180,\n2,\n2+360)}
			in (\n1:\r) arc(\n1:\n3:\r);
	}}]
	\def\r{1}

	%circle
	\path
	(0:\r) node[circle, draw=black, inner sep=1.8pt, name path=w3] (w3) {$w_3$}
	(72:\r) node[circle, draw=black, inner sep=1.8pt, name path=w2] (w2) {$w_2$}
	(144:\r) node[circle, draw=black, inner sep=1.8pt, name path=w1] (w1) {$w_1$}
	(216:\r) node[circle, draw=black, inner sep=1.8pt, name path=wn] (wn) {$w_n$}
	(288:\r) node[name path=wdots] (wdots) {$\dots$};
	\begin{scope}[black,-{To[bend]}]
		\path 
		pic{circular arc=from w1 to w2}
		pic{circular arc=from w2 to w3}
		pic{circular arc=from wn to w1};
	\end{scope}
	\begin{scope}[black,{To[bend]}-]
		\path 
		pic{circular arc=from wdots to w3}
		pic{circular arc=from wn to wdots};
	\end{scope}
	
	%nodes v and u_i
	\node[circle, draw=black, inner sep=4pt] (v) [left=of w1] {$v$};
	\node[circle, draw=black, inner sep=2.3pt] (u1) [left=of v] {$u_1$};
	\node[circle, draw=black, inner sep=2.3pt] (u2) [left=of u1] {$u_2$};
	\node (udots) [left=0.5cm of u2] {$\dots$};
	\node[circle, draw=black, inner sep=2.3pt] (un) [left=0.5cm of udots] {$u_n$};
	
	%external assets
	\node[rectangle, draw=black] (bv) [above=0.1cm of v] {$1$};
	
	%edges of v and u_i
	\draw[->] (v) to (w1);
	\draw[->] (v) to (u1);
	\draw[->] (u1) to (u2);
	\draw[->] (u2) to (udots);
	\draw[->] (udots) to (un);
	
\end{tikzpicture}
}
    \caption{A min-clearing game without a social optimum for utilitarian social welfare, for all \TM\ strategies. All edges have unit weight and the external assets of $v$ are $b_v=1$.}
    \label{fig:min-clearing-undefined-opt}
\end{figure}
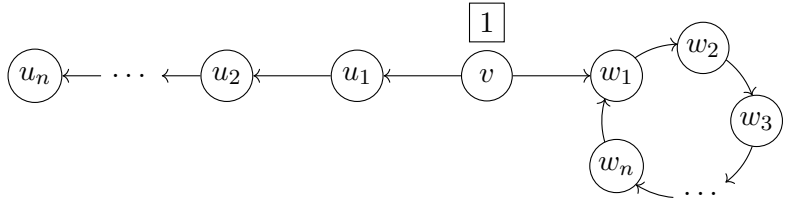

Interestingly, we can obtain positive results towards existence of strong equilibria. With min-clearing, observe that a profitable coalitional improvement step $(\veca, \veca')$ requires at least one player $v \in C$ with strictly positive assets with respect to $\veca$. Otherwise, due to min-clearing, no flow among the agents in $C$ can evolve. 
%In our proof, we construct coalitional improvement steps such that the flow increase due to the deviation originates from the player $v \in C$ with strictly positive assets in $\veca$.
%

\begin{lemma}
    \label{lem:min-clearing-SE-lem1}
    Every strategy profile in a min-clearing game with Pareto-optimal clearing state is a strong equilibrium.
\end{lemma}

\begin{proof}
    Suppose $\veca$ is a given \TM\ strategy profile with Pareto-optimal minimal clearing state $\vecf$. We assume that $\veca$ forms no strong equilibrium, therefore there exists a coalition $C \subseteq V$ where every player $v_i \in C$ strictly profits by deviating to a strategy profile $\veca'$. Hence, the players perform a coalitional improvement step implying $u_{v_i}((\veca'_C,\veca_{-C}),\vecf')>u_{v_i}(\veca,\vecf)$ for every ${v_i} \in C$. Consider player $v_i \in C$. As her external assets remain unchanged and the strategy $\veca_v$ is monotone, at least one incoming edge $e \in E^-(v_i)$ has strictly more flow. Recall that regarding minimal clearing all flow must originate from external assets. Therefore, there exists a path to $v_i$ from some player $v_j \in C$ with strictly positive assets. We denote the set of all such players by $A = \{v_j \in V \mid u_{v_j}(\veca, \vecf)>0\}$.
    %To compute the set $A_{v_i}$, we perform a modified depth-first search, operating only over incoming edges with increased money flow. 
    Now consider player $v_j \in A$. By the same arguments as above, the flow of at least one edge $e \in E^-(v_j)$ must be increased. Let $v_k$ be the source of $e$. In case $v_k\in C$, the same argument applies. On the other hand, if $v_k \notin C$ she only increases flow on any outgoing edge, if she has larger assets herself due to monotonicity. Therefore, at least one incoming edge $e \in E^-(v_k)$ has more flow regarding the new strategy profile $(\veca'_C, \veca_{-C})$. Since all flow must be initiated by a player in $A$, we can repeat this argument until we again reach some player in $A$. If this player is again $v_j$, there exists a cycle of edges with increased flow. Otherwise, when a new player $v_l$ was found we operate as for $v_j$. Because the set $A$ is finite, this process must terminate in a cycle $K$. Hence, there must be a cycle $K$ including a player with $u_v(\veca,\vecf) > 0$, along which the flow of money can be Pareto-increased. As a consequence, if a strategy profile has a Pareto-optimal clearing state $\vecf$, then there is no deviation that can increase the money circulation on any such cycle. Hence, the profile must be a strong equilibrium.
\end{proof}

For a wide range of strategies, we can extend our result and show that, for every initial strategy profile, there exists a sequence of coalitional improvement steps that end in a strong equilibrium. We consider games in which the set of strategies $\calM$ that has a consistency property. Consider a strategy $\veca_v \in \calM$ for a player $v$. Suppose we reduce all thresholds to $\hat{\tau}_e = a_e(x)$ for each $e \in E^+(v)$ and some value $0 \le x \le T_v = \sum_{e \in E^+(v)} \tau_e$. Then there should be another strategy $\hat{\veca}_v \in \calM$ such that $\hat{\veca}$ results in assignments for each $y \le x$ that are consistent with $\veca$.

Formally, we call a set $\calM$ of strategies \emph{reduction consistent} if for every vector $\vectau_v$ of thresholds, every suitable strategy $\veca_v$ (with $a_e(T_v) = \tau_e$ for all $e \in E^+(v)$) and every vector $\hat{\vectau}_v$ with $\hat{\tau}_e = a_e(\hat{T}_v)$ for some $\hat{T}_v \in [0, T_v]$, there exists a strategy $\hat{\veca}$ such that the payments $\hat{a}_e(x) = a_e(x)$, for all $e \in E^+(v)$ and all $x \in [0,\hat{T}_v]$.

\paragraph{Examples.}
Threshold edge-ranking strategies are reduction consistent. To observe this, consider a set of thresholds $\vectau_v$ and an edge-ranking strategy $\veca_v$ that allocates payments towards outgoing edges based on an ordering $\pi_v$. For any $x \in [0,T_v]$, there is at most one edge $e$ that is paid partially, i.e., $0 < a_e(x) < \tau_e$. Now for a given $\hat{T}_v \in [0,T_v]$, reduce all thresholds with $\tau_e \geq a_e(\hat{T}_v)$ to $\hat{\tau}_e = a_e(\hat{T}_v)$. Then with the same ordering $\pi_v$ in $\hat{\veca}$, the edges will be paid to most $a_e(x)$, sequentially in the same order as in $\veca$. As such, $\hat{a}_e(x) = a_e(x)$ for all $x \le \hat{T}_v$.

As another example, consider proportional strategies. Consider a set of thresholds $\vectau_v$ and a proportional strategy $\veca_v$. Recall that $\veca$ assigns payments to all edges $e \in E^+(v)$ in proportion to their thresholds $\tau_e$. If we reduce $\hat{\tau}_e = a_e(\hat{T}_v)$ for some $\hat{T}_v \in [0,T_v]$, then the proportional strategy $\hat{\veca}$ w.r.t.\ thresholds $\hat{\vectau}$ adheres to the same proportions
\[
    \frac{\hat{\tau}_e}{\sum_{e \in E^+(v)} \hat{\tau}_e} = \frac{a_e(x)}{\sum_{e \in E^+(v)} a_e(x)}  = \frac{\tau_e}{\sum_{e \in E^+(v)} \tau_e}\enspace.
\]
As such, for every $x \in [0,\hat{T}_v]$ we have $\hat{a}_e(x) = a_e(x)$.\\

Reduction consistent \TM\ strategies are \TM\ strategies with reduction consistent set $\calM$. For every game with such strategies, we construct sequences of coalitional improvement steps similar to max-clearing. However, the proof construction for max-clearing in Lemma~\ref{lem:max-clearing-SE-lem2} does not directly apply. Intuitively, when changing the strategies along a cycle, too much flow could end up being assigned to edges outside the cycle, and hence the intended flow along the cycle is not obtained in a minimal clearing state.

\begin{theorem}
    \label{thm:min-clearing-SE}
    Every min-clearing game with reduction consistent threshold-${\calM}$ strategies is coalitional weakly acyclic. For every initial strategy profile, there is a sequence of coalitional improvement steps that ends in a strong equilibrium. The sequence requires a number of steps that is polynomial in $|E|$.
\end{theorem}

%\begin{lemma}
%    \label{lem:min-clearing-SE-lem2}
%    Consider any profile $\veca$ with clearing state $\vecf$ which is not Pareto-optimal. There is a coalitional improvement step $(\veca, \veca')$ such that (1) the clearing state $\vecf'$ in $\veca'$ Pareto-dominates $\vecf$ and (2) $\vecf'$ saturates at least one more edge than $\vecf$.
%\end{lemma}

%\begin{comment}
\begin{proof}
    Existence of a strong equilibrium follows by Lemma \ref{lem:min-clearing-SE-lem1}. It remains to show that for any profile $\veca$ with clearing state $\vecf$ which is not Pareto-optimal there is a coalitional improvement step $(\veca, \veca')$ such that (1) the clearing state $\vecf'$ in $\veca'$ Pareto-dominates $\vecf$ and (2) $\vecf'$ saturates at least one more edge than $\vecf$.
    
    Given strategy profile $\veca$, consider clearing state $\vecf$. By assumption, there is another strategy profile $\tilde{\veca}$ with minimal clearing state $\tilde{\vecf}$ that Pareto dominates $\vecf$. Since all feasible flows in clearing games can be interpreted as circulation flows~\cite[Proposition 2]{BertschingerHS20}, the difference $\delta_e = \tilde{f}_e - f_e \ge 0$ is again a (non-zero) circulation flow. Consider the set of edges with circulation flow $E_{>0} = \{ e \in E \mid \delta_e > 0\}$ and the set of incident nodes $V_{\ge 0} = \{ u,v \in V \mid e = (u,v) \in E_{> 0}\}$. There must be at least one node $v^* \in V_{>0}$ such that $u_{v^*}(\tilde{\veca},\tilde{\vecf}) > \sum_{e \in E^+(v^*)} \delta_e$. Otherwise, $\tilde{f}_e = \delta_e$ for all $e \in E_{>0}$, and we could obtain a smaller clearing state in $\tilde{\veca}$ setting $\tilde{f}_e = 0$ for all $e \in E_{\ge 0}$. Note that node $v^*$ must have $u_{v^*}(\veca,\vecf) = u_{v^*}(\tilde{\veca},\tilde{\vecf}) - \sum_{e \in E^+(v^*)} \delta_e > 0$. As such, $\vecf$ must admit a cycle $K$ such that $f_e < c_e$ for all $e \in E(K)$ and $K$ contains node $v^*$ with $u_{v^*}(\veca,\vecf) > 0$. 

    We construct a coalitional improvement step as follows.
    First, compute the set $A=\{v_i \in V \mid u_{v_i}(\veca,\vecf)>0\}$ of players with strictly positive utility. Next, choose a node $v_i \in A$ and determine all simple cycles containing $v_i$ via depth-first search over the set $E' = \{ e \in E \mid f_e < c_e\}$ of non-saturated edges. If no cycle containing $v_i$ is found\footnote{At least one cycle must exist for at least one $v_i \in A$ by the discussion above.}, proceed to the next player from $A$. Otherwise, let $K'$ denote any one of the cycles found.
    We first perform an intermediate step for all players $v$ in $V(K')$, and determine the highest index $l$ such that $f_e \leq \tau^{(l)}_e$, for every outgoing edge $e \in E^+(v)$. Then, all players $v \in V(K')$ deviate to a \TM\ strategy with thresholds $\hat{\tau}_e^{(i)}=\tau_e^{(i)}$, for $i < l$ and $\hat{\tau}_e^{(l)}=f_e$, for all $e \in E^+(v)$. Choose all strategies $\hat{\veca}_v^{(1)}, \hat{\veca}_v^{(2)},\dots, \hat{\veca}_v^{(l+1)} \in \calM$ such that the resulting flow $\hat{f}_e$ equals the initial flow $f_e$ for each edge. Because all strategies $\veca_v$ are reduction consistent there exists at least one such strategy.
    We now modify the strategy $\hat{\veca}$ to increase circulation on $K'$.
    Let $\delta = \min_{e \in E(K')} c_e - \hat{f}_e$ be the minimal residual capacity of all edges in the cycle. For every $v_j \in V(K')$ and edge $e_j \in E(K')$,  
    For every edge  $e_j \in E(K')$ in the cycle, define the thresholds as $\hat{\tau}_{e_j}^{(l+1)}=\hat{\tau}_{e_j}^{(l)}+\delta$ for $l>0$ and $\hat{\tau}_{e_j}^{(1)}=\delta$. 
    Regarding all other edges $e' \in E^+(v_j) \setminus \{e_j\}$ , we set $\hat{\tau}_{e'}^{(l+1)}=\hat{\tau}_{e'}^{(l)}$ for $l>0$ and $\hat{\tau}_{e'}^{(1)}=0$. Further, define $\hat{a}_e^{(l+1)}=\hat{a}_e^{(l)}$ and arbitrarily choose $\hat{a}_e^{(1)}$ for all $e \in E^+(v_j)$. 
    Let $\veca'$ denote the resulting strategy profile. Clearly, every player $v_i \in A$ still has strictly positive assets with respect to $\veca'$. Thus, at least one player $v_j \in V(K')$ initiates flow on cycle $K'$. By construction of the first vector of thresholds $\vectau_v^{(1)}$, the flow on $K'$ is first raised to $\delta$. Afterwards, each player has the same assets available as with respect to $\veca$. From the next step on, players distribute their payments exactly as before for $\veca$. 
    Since $\delta>0$, all players in the cycle strictly profit from the deviation while the flow on all other edges remains unchanged. Moreover, the pair $(\veca'_{V(K')},\veca_{-V(K')})$ forms a coalitional improvement step for coalition $V(K')$.
\end{proof}
%\end{comment}
    
The improvement step saturates one more edge in the clearing state. Hence, any sequence of such steps is limited in length by $|E|$. Note that each improvement step adds one vector of thresholds to the strategy of every player. In contrast, for max-clearing games the strong equilibrium can be obtained using only a single set of thresholds for each player. It is an interesting open problem if a similar property can be shown for min-clearing games with reduction consistent strategies. Moreover, it is unclear whether a sequence of improvement steps ending in a strong equilibrium exists for strategies without said property.

The sequence of coalitional improvement steps to a strong equilibrium can even be performed for threshold proportional strategies, when given an initial profile \emph{along with its minimal clearing state}. For constructing the improvement step, it is sufficient to recognize whether there is a non-saturated cycle with at least one player having strictly positive assets. We obtain $\veca'$ as described above, and we obtain the minimal clearing state of $\veca'$ adaptively by increasing the flow along the cycle until the first edge becomes saturated.

The non-existence of a social optimum in Example~\ref{ex:noOpt} arises from a continuity problem. For every flow $f_e$ on edge $e=(v,w_1)$ there exists some flow $\vecf'$ with $f'_e< f_e$ and strictly higher welfare. A simple way to resolve this problem is to restrict attention to a finite set of strategies for each player, i.e., a finite set of possible threshold vectors and a finite set $\calM$ of strategies for same-priority installments. Then, an optimal strategy profile always exists. A natural example enjoying this discretization are \emph{coin-ranking} strategies studied in~\cite{BertschingerHS20}, where flow is split into ``coins'', i.e., units of 1. Each player then decides on an order in which the coins will be distributed to her outgoing edges. Note that a coin-ranking strategy can be interpreted as threshold edge-ranking strategy with $k = b_v + \sum_{e \in E^-(v)} c_e$, such that for every $i<k$ we have $\tau_e^{(i+1)}=\tau_e^{(i)}+1$ for exactly one edge $e\in E^+(v)$ and $\tau_{e'}^{(i+1)}=\tau_{e'}^{(i)}$ for all other edges $e'\in E^+(v) \setminus \{e\}$. Hence, coin-ranking strategies are threshold edge-ranking strategies restricted to a finite number of possible threshold vectors (and a finite set $\calM$).

It is easy to see that the arguments in Theorem \ref{thm:min-clearing-SE} apply to coin-ranking strategies, since the clearing state $\vecf$ will be an integral flow. Hence, for every min-clearing game with coin-ranking strategies a strong equilibrium exists and can be computed efficiently. However, the following theorem shows that computing the optimal strategy profile in such games is \classNP-hard.

\begin{theorem}
    \label{thm:min-clearing-coin-NPC}
	For min-clearing games with threshold edge-ranking strategies and unit thresholds, it is \classNP-hard to compute a strategy profile that maximizes utilitarian social welfare.
\end{theorem}

\begin{proof}
    Formally, we will present a class of games such that it is \classNP-complete to decide if there is a strategy profile with social welfare at least $2n+1$. The reduction is from the \DHamilPath\ problem. Given an instance of this problem, i.e., a directed graph $G=(V,E)$, we construct a min-clearing game with graph $G'=(V',E')$ as follows. First, add a node $s$ with $b_s=1$ to $V'$. This will be the only player with external assets and, hence, the origin of the total flow in the network. For every node $v_i \in V$, construct two nodes $v_i^-$ and $v_i^+$ and add the edges $(s, v_i^-)$ and $(v_i^-, v_i^+)$ to $E'$. Finally, for every edge $(v_i,v_j)\in E$ include $(v_i^+, v_j^-)$ in $E'$ and set all edge weights to 1. It is easy to see that this construction can be performed in $O(|V|+|E|)$ time.
	
	We claim that a Hamiltonian path in $G$ exists if and only if there is a strategy profile with social welfare of at least $2n+1$ in $G'$.
	First, assume there exists a Hamiltonian path in $G$. We rename the nodes in $V$ and analogously in $V'$ for the path to be given by $(v_1,\dots,v_n)$. By means of this path, the strategy profile $\veca$ will be constructed. Player $s$ introduces flow by paying her debt to $v_1$, i.e., $\tau_{(s, v_1^-)}^{(1)}=1$ and $\tau_{e'}^{(1)}=0$ for all $e' \in E^+(s)\setminus \{(s, v_1^-)\}$. The remaining thresholds can be chosen arbitrarily. Analogous, for every player $v_i^+$ with $i<n$ we prioritize edge $(v_i^+, v_{i+1}^-)$, hence, we define $\tau_{(v_i^+, v_{i+1}^-)}^{(1)}=1$ and $\tau_{e'}^{(1)}=0$ for all other edges. Note that players $v_i^-$ only have one outgoing edge which predetermines their strategies. In the resulting flow $\vecf$, every player $v_i^+$ and $s$ hold at most assets of one. With the exception of one player, the incoming flow of all players $v_i^-$ is limited to one. If $v_n^+$ has an outgoing edge, exactly one player $v_i^-$ can own assets of two. The described strategy profile results in a minimal clearing state where all prioritized edges are saturated. In conclusion, every player $v_i^-$ and $v_i^+$ and $s$ hold assets of at least one, implying the social welfare is at least $2n+1$.
	
	Now assume there exists a strategy profile $\veca$ with social welfare at least $2n+1$.
	W.l.o.g. $s$ pays off its liabilities to $v_1^-$, implying assets of one for players $v_1^-$ and $v_1^+$. Let $(v_1^+,v_2^-)$ be the first edge $v_1^+$ pays, otherwise we rename the nodes accordingly. Thus, the assets of $v_2^-$ and $v_2^+$ are increased to 1. We denote the player $v_2^+$ transfers her incoming flow to by $v_i^-$. Social welfare of $2n+1$ can only be obtained, if player $v_i^-$ is currently without assets. Otherwise, this player would hold assets of at least two with only one outgoing edge. Hence, the flow would yield social welfare of at most $2k+2$ for $2k$ visited players. By iteratively applying this argument, a social welfare of at least $2n+1$ can only be achieved if there exists a path including all nodes, or in other words a Hamiltonian path.
\end{proof}

\section{Clearing Games with Fixed Seniorities}\label{sec:fixed-senorities}
In this section, we consider clearing games with fixed seniorities, i.e., when thresholds are fixed exogenously. In this case, a player with fixed thresholds is only allowed to choose appropriate strategies $a_e^{(i)}$. We will concentrate on the simpler case of fixed \emph{partition} strategies. A bank with a fixed partition strategy is obligated to pay debt of edges completely in a given order with ties. A strategic choice arises only in case of ties, i.e., for edges of the same priority, where the bank can strategically choose $a_e^{(i)}$ to pay for the edges with seniority $i$.

In the remainder of this section, we consider games in which some agents have fixed partition-$\calM$ strategies, while others are free to choose threshold-$\calM$ strategies, for any suitable set $\calM$. For the fixed partitions, we specify the priority class of an edge. Our results imply that the positive properties for \TM\ strategies do not persist if partitions are fixed for some players. Notably, these conditions are not implied by existing constructions in~\cite{BertschingerHS20, Kanellopoulos21}.

We start by proving $\classNP$-hardness for computation of the optimal strategy profile in both min- and max-clearing games.

\begin{theorem}\label{thm:sen-compute-OPT}
    Consider max-clearing and min-clearing games with \TM\ strategies and $k\geq 2$. For fixed seniorities, it is \classNP-hard to compute an optimal strategy profile with respect to social welfare.
\end{theorem}
%\begin{comment}
\begin{proof}
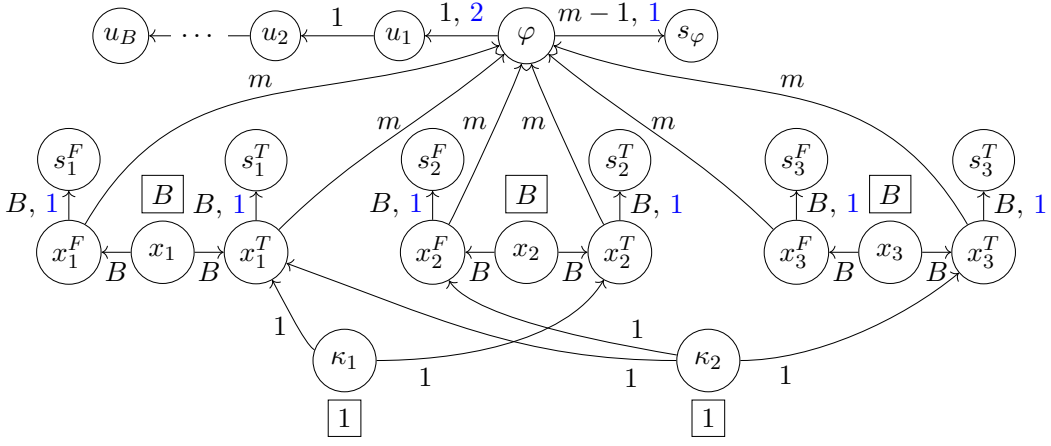
\begin{figure}[t]
\centering
\resizebox{!}{0.35\textwidth}{
\begin{tikzpicture}
    %Gadget 1
    \node[circle, draw=black, inner sep=4.3pt] (1) at (0,0) {$x_1$};
    \node[rectangle, draw=black] (b1) [above=0.1cm of 1] {$B$};
    \node[circle, draw=black, inner sep=2.5pt] (1t) [right=0.4cm of 1] {$x_1^T$};
    \node[circle, draw=black, inner sep=2.5pt] (1f) [left=0.4cm of 1] {$x_1^F$};
    \node[circle, draw=black, inner sep=2.5pt] (s1t) [above=0.4cm of 1t] {$s_1^T$};
    \node[circle, draw=black, inner sep=2.5pt] (s1f) [above=0.4cm of 1f] {$s_1^F$};
    \draw[->] (1) to node[below, midway] {$B$} (1t);
    \draw[->] (1) to node[below, midway] {$B$} (1f);
    \draw[->] (1t) to node[left, midway] {$B$, \textcolor{blue}{1}} (s1t);
    \draw[->] (1f) to node[left, midway] {$B$, \textcolor{blue}{1}} (s1f);
    
    %Gadget 2
    \node[circle, draw=black, inner sep=4.3pt] (2) at (5,0) {$x_2$};
    \node[rectangle, draw=black] (b2) [above=0.1cm of 2] {$B$};
    \node[circle, draw=black, inner sep=2.5pt] (2t) [right=0.4cm of 2] {$x_2^T$};
    \node[circle, draw=black, inner sep=2.5pt] (2f) [left=0.4cm of 2] {$x_2^F$};
    \node[circle, draw=black, inner sep=2.5pt] (s2t) [above=0.4cm of 2t] {$s_2^T$};
    \node[circle, draw=black, inner sep=2.5pt] (s2f) [above=0.4cm of 2f] {$s_2^F$};
    \draw[->] (2) to node[below, midway] {$B$} (2t);
    \draw[->] (2) to node[below, midway] {$B$} (2f);
    \draw[->] (2t) to node[right, midway] {$B$, \textcolor{blue}{1}} (s2t);
    \draw[->] (2f) to node[left, midway] {$B$, \textcolor{blue}{1}} (s2f);
    
    %Gadget 3
    \node[circle, draw=black, inner sep=4.3pt] (3) at (10,0) {$x_3$};
    \node[rectangle, draw=black] (b3) [above=0.1cm of 3] {$B$};
    \node[circle, draw=black, inner sep=2.5pt] (3t) [right=0.4cm of 3] {$x_3^T$};
    \node[circle, draw=black, inner sep=2.5pt] (3f) [left=0.4cm of 3] {$x_3^F$};
    \node[circle, draw=black, inner sep=2.5pt] (s3t) [above=0.4cm of 3t] {$s_3^T$};
    \node[circle, draw=black, inner sep=2.5pt] (s3f) [above=0.4cm of 3f] {$s_3^F$};
    \draw[->] (3) to node[below, midway] {$B$} (3t);
    \draw[->] (3) to node[below, midway] {$B$} (3f);
    \draw[->] (3t) to node[right, midway] {$B$, \textcolor{blue}{1}} (s3t);
    \draw[->] (3f) to node[right, midway] {$B$, \textcolor{blue}{1}} (s3f);
    
    %Clause 1
    \node[circle, draw=black, inner sep=4.3pt] (k1) at (2.5,-1.5) {$\kappa_1$};
    \node[rectangle, draw=black] (bk1) [below=0.1cm of k1] {$1$};
    \draw[->] (k1) to[in=300, out=160, looseness=0.3] node[left, pos=0.4] {1} (1t);
    \draw[->] (k1) to[in=240, out=0, looseness=0.8] node[below, pos=0.2] {1} (2t);
    
    %Clause 2
    \node[circle, draw=black, inner sep=4.3pt] (k2) at (7.5,-1.5) {$\kappa_2$};
    \node[rectangle, draw=black] (bk2) [below=0.1cm of k2] {$1$};
    \draw[->] (k2) to[in= 340, out=180, looseness=1] node[below, pos=0.1] {1} (1t);
    \draw[->] (k2) to[in=300, out=170, looseness=0.5] node[above, pos=0.2] {1} (2f);
    \draw[->] (k2) to[in=220, out=0, looseness=0.8] node[below, pos=0.2] {1} (3t);
    
    %Phi, path graph
    \node[circle, draw=black, inner sep=4.3pt] (phi) at (5,3) {$\varphi$};
    \node[circle, draw=black, inner sep=2.5pt] (sphi) [right=1.5cm of phi] {$s_{\varphi}$};
    \node[circle, draw=black, inner sep=2.5pt] (u1) [left= of phi] {$u_1$};
    \node[circle, draw=black, inner sep=2.5pt] (u2) [left= of u1] {$u_2$};
    \node (udots) [left=0.3cm of u2] {$\dots$};
    \node[circle, draw=black, inner sep=2.5pt] (uB) [left=0.3cm of udots] {$u_B$};
    \draw[->] (1t) to[in=220, out=45, looseness=1] node[above, midway] {$m$} (phi);
    \draw[->] (1f) to[in=200, out=60, looseness=1.1] node[above, midway] {$m$} (phi);
    \draw[->] (2t) to[in=280,out=120,looseness=0.2] node[left, pos=0.6] {$m$} (phi);
    \draw[->] (2f) to[in=260,out=60,looseness=0.2] node[left, pos=0.6] {$m$} (phi);
    \draw[->] (3t) to[in=340,out=120,looseness=1.1] node[above, midway] {$m$} (phi);
    \draw[->] (3f) to[in=320,out=135,looseness=1] node[above, midway] {$m$} (phi);
    \draw[->] (phi) to node[above, midway] {$m-1$, \textcolor{blue}{1}} (sphi);
    \draw[->] (phi) to node[above, midway] {1, \textcolor{blue}{2}} (u1);
    \draw[->] (u1) to node[above, midway] {1} (u2);
    \draw[-] (u2) to (udots);
    \draw[->] (udots) to (uB);
\end{tikzpicture}
}
    \caption{Example construction of a game in the proof of Theorem~\ref{thm:sen-compute-OPT}. Black edge labels indicate edge weights while blue labels indicate seniorities.}
    \label{fig:sen-compute-OPT}
\end{figure}
    %\changed{
    We prove the statement by reduction from \tSAT . For a given instance, let $\varphi$ denote the Boolean formula with variables $x_1,x_2,\dots,x_n$ and clauses $\kappa_1,\kappa_2,\dots,\kappa_m$. We construct a clearing game with an acyclic graph $G=(V, E)$, hence, all flow originates from external assets. For this reason,
    in this game minimal and maximal clearing states coincide for every strategy profile. We construct the game as follows.
    
    First, we define a gadget for every variable $x_i$ consisting of the players $x_i, x_i^T, x_i^F, s_i^T$ and $s_i^F$.
    Player $x_i$ holds external assets of $B\gg m$, where $B$ is polynomial in $n$ and $m$. Further, $x_i$ owes assets of $B$ to both $x_i^T$ and $x_i^F$, i.e., $(x_i, x_i^T), (x_i,x_i^F)\in E$. Additionally, we add the edges $(x_i^T,s_i^T)$ and $(x_i^F,s_i^F)$ with seniority 1 and weight $B$ each. %Hence, these edges must be paid in full before $x_i^T$ and $x_i^F$ can begin to pay off any other liabilities.
    For every clause $\kappa_j$, we include player $\kappa_j$ with external assets of 1. If variable $x_i$ appears as $\neg x_i$ (or $x_i$) in clause $\kappa_j$, the unit-weight edge $(\kappa_j,x_i^F)$ (or $(\kappa_j, x_i^T)$) is added to $E$. %Otherwise, if it appears as $x_i$ in $\kappa_j$, we include $(\kappa_j,x_i^T)$ with unit-weight. 
    Further, we add player $\varphi$ and edges $(x_i^T,\varphi)$ and $(x_i^F,\varphi)$ with edge weight $m$ and seniority 2 for every $i$.
    Player $\varphi$ has liabilities of $m-1$ to sink $s_\varphi$ with seniority 1 and liabilities of 1 with seniority 2 to a linear graph consisting of $u_1,u_2,\dots,u_B$.  
    An example construction for the instance $(x_1 \vee x_2) \wedge (x_1 \vee \neg  x_2 \vee x_3)$ is depicted in Figure \ref{fig:sen-compute-OPT}. 
    
    Intuitively, when player $x_i$ pays off all assets towards $x_i^T$ (or $x_i^F$), we interpret the literal $x_i$ (or $\neg x_i$) in $\varphi$ to be assigned with TRUE. The players $x_i^T$ and $x_i^F$ are obligated to first pay off all debt towards a sink. Therefore, one of these players can make payments outside the gadget only if the literal associated with her is true. We call such a player \emph{unlocked}. % if she has fully paid off the sink. 
    Note that by the choice of $B$, it is impossible to unlock both $x_i^T$ and $x_i^F$ with flow of $B+m$ which is the maximal flow achieved when all $\kappa_j$ pay their assets towards the same gadget.
    Assume $\kappa_j$ to settle the debt to an unlocked player. By definition, she can forward the payments from $\kappa_j$ to $\varphi$. In contrast, a \emph{locked} player would forward the payments to the respective sink. Whenever player $\varphi$ receives incoming flow of $k$, at least $k$ players $\kappa_j$ pass their assets to an unlocked player. Thus, the assets of $\varphi$ can be interpreted as the number of satisfied clauses. Thus, whenever there is strictly positive flow along the linear graph, all clauses are interpreted as satisfied.
   
    For every formula $\varphi$, the number of nodes and edges in the constructed game is linear in the number of variables $n$ and number of clauses $m$. Hence, the construction can be performed in polynomial time.
    
    We state that there exists a satisfying assignment of $\varphi$ if and only if the social welfare of the optimal strategy profile is $(3n+1)B+4m-1$.
    First, assume there exists a satisfying assignment of $\varphi$. We then define a strategy profile $\veca$ along the lines of this assignment. For every variable $x_i$ that is assigned TRUE, $x_i$ prioritizes the edge $(x_i, x_i^T)$, i.e., $\tau_{(x_i,x_i^T)}^{(1)}=1$ and $\tau_{(x_i,x_i^F)}^{(1)}=0$. In the case, when $x_i$ is assigned FALSE we proceed analogously and prioritize the edge towards $x_i^F$. For each clause $\kappa_j$, choose one arbitrary clause-fulfilling literal $l_j$ and prioritize the edge $(\kappa_j,l_j)$ as described before. Then in each variable gadget, the social welfare of at least $3B$ is achieved. Also, every $\kappa_j$ pays her assets towards an unlocked player, leading to incoming flow of $m$ for $\varphi$. Thus, $\varphi$ can settle all debts leading to assets of $m-1$ for $s_{\varphi}$ and assets of 1 for every player $u_i$ in the path graph. In total, the social welfare is $(3n+1)B+4m-1$.
    
    For the other direction, assume there exists a strategy profile $\veca$ yielding social welfare of $(3n+1)B+4m-1$. In the constructed game, $n$ players each hold external assets of $B$, which are paid in full via $x_i^T$ and $x_i^F$ to sinks $s_i^T$ and $s_i^F$. This generates social welfare of $3Bn$. The social welfare can only be increased by another $B$, if flow of 1 exists on the path graph $u_1,\dots,u_B$. Since this flow must be initiated by $\varphi$, her incoming flow must be at least $m$. Thus, every $\kappa_j$ pays her external assets via an unlocked player towards $\varphi$.
    %}
\end{proof}
%\end{comment}

For clearing games and \TM\ strategies, the strategy profile with Pareto-optimal clearing state has been shown to be a strong equilibrium and, in case of max-clearing, even efficiently computable. By fixing seniorities, however, we do not only impact the computational complexity of the optimal strategy profile but also the existence of stable states. In fact, the subsequent theorem shows that there are max-clearing games with \TM\ strategies without pure Nash equilibria.

\begin{proposition}
    \label{prop:max-clearing-sen-NE}
    For \TM\ strategies with any $k\geq 2$, there exists a max-clearing game with fixed seniorities that has no pure Nash equilibrium.
\end{proposition}
%\begin{comment}
\begin{proof}
\begin{figure}[t]
\begin{subfigure}[l]{0.48\textwidth}
    \centering
    \resizebox{!}{0.8 \textwidth}{
    \begin{tikzpicture}
		\pgfmathsetmacro{\shift}{1ex}
		\node[circle, draw=black, inner sep=2.5pt] (1) at (210:3) {$v_1$};
		\node[circle, draw=black, inner sep=2.5pt] (4) at (330:3) {$v_4$};
		\node[circle, draw=black, inner sep=2.5pt] (7) at (90:3) {$v_7$};
		\node[circle, draw=black, inner sep=2.5pt] (2) at ($(1)!0.25!(4)$) {$v_2$};
		\node[circle, draw=black, inner sep=2.5pt] (3) at ($(1)!0.5!(4)$) {$v_3$};
		\node[circle, draw=black, inner sep=2.5pt] (5) at ($(4)!0.25!(7)$) {$v_5$};
		\node[circle, draw=black, inner sep=2.5pt] (6) at ($(4)!0.5!(7)$) {$v_6$};
		\node[circle, draw=black, inner sep=2.5pt] (8) at ($(7)!0.25!(1)$) {$v_8$};
		\node[circle, draw=black, inner sep=2.5pt] (9) at ($(1)!0.5!(7)$) {$v_9$};
		\node[circle, draw=black, inner sep=1pt] (10) at ($(3)!0.5!(9)$) {$v_{10}$};
		\node[circle, draw=black, inner sep=1pt] (11) at ($(3)!0.5!(6)$) {$v_{11}$};
		\node[circle, draw=black, inner sep=1pt] (12) at ($(6)!0.5!(9)$) {$v_{12}$};
		
		%outgoing edges v_1
		\draw[transform canvas={yshift=0.5*\shift},->] (1) -- (2);
		\draw[transform canvas={yshift=-0.5*\shift},->](1) -- (2);
		\draw[transform canvas={xshift=0.7*\shift, yshift=-0.4*\shift},->] (1) -- (9);
		\draw[->](1) -- (9);
		\draw[transform canvas={xshift=-0.7*\shift, yshift=0.4*\shift},->](1) -- (9);

		%outgoing edges v_2
		\draw[->] (2) to node[above, midway] {\small \textcolor{blue}{1}} (3);
		\draw[->] (2) to[bend left] node[below, midway] {\small \textcolor{blue}{2}} (1);
		
		%outgoing edges v_3
		\draw[->] (3) to node[right, midway] {\small \textcolor{blue}{1}} (10);
		\draw[->] (3) to node[right, midway] {\small \textcolor{blue}{2}} (11);
		\draw[->] (3) to[bend right=36] node[below=-2pt, midway] {\small \textcolor{blue}{3}} (4);
		\draw[->] (3) to[bend right=65] node[below, midway] {\small \textcolor{blue}{4}} (4);
		
		%outgoing edges v_4
		\draw[transform canvas={xshift=0.5*\shift, yshift=0.25*\shift},->] (4) -- (5);
		\draw[transform canvas={xshift=-0.5*\shift, yshift=-0.25*\shift},->](4) -- (5);
		\draw[transform canvas={yshift=-0.8*\shift},->] (4) --(3);
		\draw[->](4) -- (3);
		\draw[transform canvas={yshift=0.8*\shift},->](4) -- (3);
		
		%outgoing edges v_5
		\draw[->] (5) to node[right, midway] {\small \textcolor{blue}{1}} (6);
		\draw[->] (5) to[bend left] node[right, midway] {\small \textcolor{blue}{2}} (4);
		
		%outgoing edges v_6
		\draw[->] (6) to node[right, midway] {\small \textcolor{blue}{1}} (11);
		\draw[->] (6) to node[above, midway] {\small \textcolor{blue}{2}} (12);
		\draw[->] (6) to[bend right=36] node[right, midway] {\small \textcolor{blue}{3}} (7);
		\draw[->] (6) to[bend right=65] node[right, midway] {\small \textcolor{blue}{4}} (7);
		
		%outgoing edges v_7
		\draw[transform canvas={xshift=-0.5*\shift, yshift=0.25*\shift},->] (7) -- (8);
		\draw[transform canvas={xshift=0.5*\shift, yshift=-0.25*\shift},->](7) -- (8);
		\draw[transform canvas={xshift=0.7*\shift, yshift=0.4*\shift},->] (7) -- (6);
		\draw[->](7) -- (6);
		\draw[transform canvas={xshift=-0.7*\shift, yshift=-0.4*\shift},->](7) -- (6);
		
		%outgoing edges v_8
		\draw[->] (8) to node[left, midway] {\small \textcolor{blue}{1}} (9);
		\draw[->] (8) to[bend left] node[left, midway] {\small \textcolor{blue}{2}} (7);
		
		%outgoing edges v_9
		\draw[->] (9) to node[above, midway] {\small \textcolor{blue}{1}} (12);
		\draw[->] (9) to node[left, midway] {\small \textcolor{blue}{2}} (10);
		\draw[->] (9) to[bend right=36] node[left, midway] {\small \textcolor{blue}{3}} (1);
		\draw[->] (9) to[bend right=65] node[left, midway] {\small \textcolor{blue}{4}} (1);
		
		%outgoing edges v_10
		\draw[->] (10) -- ($(10)!0.84!(1)$);
		
		%outgoing edges v_11
		\draw[->] (11) -- ($(11)!0.84!(4)$);
		
		%outgoing edges v_12
		\draw[->] (12) -- ($(12)!0.84!(7)$);
	\end{tikzpicture}
	}
	\caption{A max-clearing game with \TM\ strategies for players $v_1$, $v_4$ and $v_7$ and fixed partition-$\calM$ strategies for the remaining players (blue edge labels). The game has no pure Nash equilibrium.}
	\label{fig:max-clearing-sen-NE}
\end{subfigure}
\hfill
\begin{subfigure}[r]{0.48\textwidth}
    \centering
    %\vspace{0.5cm}
    \setlength{\extrarowheight}{0.1cm}
    \begin{tabular}{c|c|c}
    	\multicolumn{3}{c}{$v_7$ : (1)}\\ \hline \hline 
    	\diag{0.1em}{0.5cm}{$v_1$}{$v_4$}& (1) & (2)\\ \hline
    	(1) & 2,2,2 & 2,3,2 \\ \hline
    	(2) & 3,2,2 & 3,0,1 \\
    \end{tabular} \vspace{0.6cm} \\
    \begin{tabular}{c|c|c}
    	\multicolumn{3}{c}{$v_7$ : (2)}\\ \hline \hline 
    	\diag{0.1em}{0.5cm}{$v_1$}{$v_4$}& (1) & (2)\\ \hline
    	(1) & 2,2,3 & 2,3,0 \\ \hline
    	(2) & 0,2,3 & 1,1,1 \\
    \end{tabular}
    \vspace{0.6cm}
    \caption{The game among players $v_1$, $v_4$ and $v_7$ has only two meaningful strategies for each player. Entries of the tables denote utility of $v_1,v_4$ and $v_7$, respectively.}
    \label{tab:max-clearing-sen-NE}
\end{subfigure}
    \caption{Max-clearing game without pure Nash equilibrium.}
\end{figure}
    Consider the game depicted in Figure \ref{fig:max-clearing-sen-NE} where no player owns external assets and all edges have unit weights. The labels of the edges indicate the seniorities. Note that only players $v_1,v_4$ and $v_7$ have a meaningful strategy choice, since the behavior of all others is essentially determined.
    
    Player $v_1$ has liabilities to $v_2$ and $v_9$ of same priority, thus her strategy choices can be divided into three distinct categories: (1) She directs all payments towards $v_2$ until her liabilities towards $v_2$ are completely payed off, (2) she directs all payments towards $v_9$ until her liabilities towards $v_9$ are completely payed off, or (3) she splits her payments between the two in some arbitrary way.
    A strategy from the last category is never profitable for the player, since part of her payments are either retained by $v_{10}$ or passed on from $v_9$ to $v_{12}$. On the other hand, if the player uses all her assets to pay off one debt, she can exploit the cycles and increase her assets depending on the strategies of the other players. Therefore, we restrict ourselves to categories (1) and (2). By symmetry, the same arguments apply to strategies of $v_4$ and $v_7$. The utilities of all resulting strategy profiles are illustrated in Figure \ref{tab:max-clearing-sen-NE}. For each combination of strategies, at least one player strictly benefits from a unilateral deviation. 
\end{proof}
%\end{comment}

In the game analyzed in the previous proof, it is never profitable for any player $v_1$, $v_4$ or $v_7$ to split her payments to different agents. A player always prefers to direct all money flow towards one player. Thus, the choices of $a_e^{(i)}$ are inconsequential and only one threshold per edge is needed. 

We can adjust the game in the proof of Proposition~\ref{prop:max-clearing-sen-NE} to reduce each fixed partition to only \emph{two} priority classes. To see this, consider a player $v \in \{v_3, v_6, v_9\}$ whose outgoing edges are divided into four seniority classes. We can add an auxiliary player $w$ and redirect all edges with seniority higher than 1 towards $w$. In particular, replace every edge $(v,u)$ with seniority $j$, with edges $(v,w)$ with seniority 2 and $(w,u)$ with seniority $j-1$. Then proceed to node $w$. In each step, the required number of seniority classes is reduced by one. We obtain an equivalent game where every fixed partition for a node contains only two classes.

The next result revisits the discussed game and proves that deciding existence of an equilibrium for a given max-clearing game with fixed seniorities is \classNP-hard.

\begin{theorem}\label{thm:max-clearing-sen-compute-NE}
    For a given max-clearing game with fixed seniorities and \TM\ strategies with any $k \geq 2$, it is \classNP-hard to decide whether a pure Nash equilibrium exists.
\end{theorem}
%\begin{comment}
\begin{proof}
    Consider the game with no Nash equilibrium described in proof of Proposition~\ref{prop:max-clearing-sen-NE}. In this game, we can generate an equilibrium by paying additional assets of 1 to player $v_1$. Then, $v_1$ can pay off her total liabilities when $v_4$ chooses strategy (2) and $v_7$ chooses (1). When playing according to this strategy profile, no player has an incentive to deviate unilaterally. The additional assets are generated by the construction described in proof of Theorem \ref{thm:sen-compute-OPT}. Recall, that we constructed a game for a given instance $\varphi$ of \tSAT. The player $u_B$ has assets of 1 if and only if there exists a satisfying assignment of $\varphi$. Hence by adding the edge $(u_B,v_1)$, the statement follows.
\end{proof}
%\end{comment}

Both results regarding max-clearing can be extended to min-clearing games. For non-existence of Nash equilibria, a game is constructed where three players are able to freely choose their strategies while the behavior of all other players is fully determined by seniorities.

\begin{proposition}
    \label{prop:min-clearing-sen-NE}
    For \TM\ strategies with any $k\geq 2$, there exists a min-clearing game with fixed seniorities that has no pure Nash equilibrium.
\end{proposition}

\begin{proof}
    \begin{figure}[t]
    \begin{subfigure}[l]{0.45\textwidth}
    \centering
    \begin{tikzpicture}
	\def\r{2.5}
	%nodes
	\node[circle, draw=black, inner sep=2.5pt] (1) at (210:\r) {$v_1$};
%	\node[rectangle, draw=black] (b1) [below left=0.1cm of 1] {\footnotesize $\delta+1$};
	\node[circle, draw=black, inner sep=2.5pt] (5) at (330:\r) {$v_5$};
%	\node[rectangle, draw=black] (b5) [below right=0.1cm of 5] {\footnotesize $\delta+1$};
	\node[circle, draw=black, inner sep=2.5pt] (9) at (90:\r) {$v_9$};
%	\node[rectangle, draw=black] (b9) [above=0.1cm of 9] {\footnotesize $\delta+1$};
	\node[circle, draw=black, inner sep=2.5pt] (2) at ($(1)!0.25!(5)$) {$v_2$};
	\node[circle, draw=black, inner sep=2.5pt] (3) at ($(1)!0.5!(5)$) {$v_3$};
	\node[circle, draw=black, inner sep=2.5pt] (4) at ($(1)!0.75!(5)$) {$v_4$};
	\node[circle, draw=black, inner sep=2.5pt] (6) at ($(5)!0.25!(9)$) {$v_6$};
	\node[circle, draw=black, inner sep=2.5pt] (7) at ($(5)!0.5!(9)$) {$v_7$};
	\node[circle, draw=black, inner sep=2.5pt] (8) at ($(5)!0.75!(9)$) {$v_8$};
	\node[circle, draw=black, inner sep=1pt] (10) at ($(9)!0.25!(1)$) {$v_{10}$};
	\node[circle, draw=black, inner sep=1pt] (11) at ($(9)!0.5!(1)$) {$v_{11}$};
	\node[circle, draw=black, inner sep=1pt] (12) at ($(9)!0.75!(1)$) {$v_{12}$};
	\node[circle, draw=black, inner sep=1pt] (13) [below=0.5cm of 2] {$v_{13}$};
	\node[circle, draw=black, inner sep=1pt] (14) [below=0.5cm of 3] {$v_{14}$};
	\node[circle, draw=black, inner sep=1pt] (15) [below=0.5cm of 4] {$v_{15}$};
	\node[circle, draw=black, inner sep=1pt, position=30:{0.5cm} from 6] (16) {$v_{16}$};
	\node[circle, draw=black, inner sep=1pt, position=30:{0.5cm} from 7] (17) {$v_{17}$};
	\node[circle, draw=black, inner sep=1pt, position=30:{0.5cm} from 8] (18) {$v_{18}$};
	\node[circle, draw=black, inner sep=1pt, position=150:{0.5cm} from 10] (19) {$v_{19}$};
	\node[circle, draw=black, inner sep=1pt, position=150:{0.5cm} from 11] (20) {$v_{20}$};
	\node[circle, draw=black, inner sep=1pt, position=150:{0.5cm} from 12] (21) {$v_{21}$};
	
	%outgoing edges v_1
	\draw[->] (1) to (2);
	\draw[->] (1) to (12);
	
	%outgoing edges v_2
	\draw[->] (2) to (3);
	\draw[->] (2) to (13);
	
	%outgoing edges v_3
	\draw[->] (3) to[in=18, out=140, looseness=1.1] (1);
	\draw[->] (3) to[in=162,out=40, looseness=1.1] (5);
	\draw[->] (3) to (14);

	%outgoing edges v_4
	\draw[->] (4) to (3);
	\draw[->] (4) to[bend right=50] (5);
	\draw[->] (4) to (15);	
	
	%outgoing edges v_5
	\draw[->] (5) to (6);
	\draw[->] (5) to (4);
	
	%outgoing edges v_6
	\draw[->] (6) to (7);
	\draw[->] (6) to (16);
	
	%outgoing edges v_7
	\draw[->] (7) to[in=282, out=160 , looseness=1.1](9);
	\draw[->] (7) to (17);
	\draw[->] (7) to[in=137,out=260, looseness=1.1] (5);
	
	%outgoing edges v_8
	\draw[->] (8) to (7);
	\draw[->] (8) to[bend right=50] (9);
	\draw[->] (8) to (18);
	
	%outgoing edges v_9
	\draw[->] (9) to (10);
	\draw[->] (9) to (8);
	
	%outgoing edges v_10
	\draw[->] (10) to (11);
	\draw[->] (10) to (19);
	
	%outgoing edges v_11
	\draw[->] (11) to[in=257, out=20, looseness=1.1] (9);
	\draw[->] (11) to[in=43, out=280, looseness=1.1] (1);
	\draw[->] (11) to (20);
	
	%outgoing edges v_12
	\draw[->] (12) to (11);
	\draw[->] (12) to[bend right=50] (1);
	\draw[->] (12) to (21);
    \end{tikzpicture}
    \end{subfigure}
    %%%%%%%%%%%%%%%%%%%%%%%%%%%%%
    \begin{subfigure}[r]{0.45\textwidth}
    \centering
    \begin{tikzpicture}
	\def\r{2.5}
	%nodes
	\node[circle, draw=black, inner sep=2.5pt] (1) at (210:\r) {$v_1$};
	\node[rectangle, draw=black] (b1) [below left=0.1cm of 1] {\footnotesize $\delta+1$};
	\node[circle, draw=black, inner sep=2.5pt] (3) at (330:\r) {$v_3$};
	\node[circle, draw=black, inner sep=1pt] (11) at (90:\r) {$v_{11}$};
	\node[circle, draw=black, inner sep=2.5pt] (2) at ($(1)!0.4!(3)$) {$v_2$};;
	\node[circle, draw=black, inner sep=1pt] (12) at ($(11)!0.6!(1)$) {$v_{12}$};
	\node[circle, draw=black, inner sep=1pt] (13) [below=1cm of 2] {$v_{13}$};
	\node[circle, draw=black, inner sep=1pt] (14) [below=1cm of 3] {$v_{14}$};
	\node[circle, draw=black, inner sep=1pt, position=150:{1cm} from 11] (20) {$v_{20}$};
	\node[circle, draw=black, inner sep=1pt, position=150:{1cm} from 12] (21) {$v_{21}$};
	
	%outgoing edges v_1
	\draw[->] (1) to node[below, midway] {\footnotesize $\delta+3$} (2);
	\draw[->] (1) to node[above, midway, sloped] {\footnotesize $\delta+2$} (12);
	
	%outgoing edges v_2
	\draw[->] (2) to node[below, midway, align=right] {\footnotesize $3$, \footnotesize \textcolor{blue}{$\delta+1$},\\ \footnotesize \textcolor{blue}{$\dots$}, \\ \footnotesize \textcolor{blue}{ $\delta+3$}} (3);
	\draw[->] (2) to node[left, midway, align=right] {\footnotesize $\delta$, \\ \footnotesize \textcolor{blue}{1, \dots, $\delta$}} (13);
	
	%outgoing edges v_3
	\draw[->] (3) to[in=18, out=140, looseness=1.1] node[above=0.2cm, right] {\footnotesize 3, \footnotesize \textcolor{blue}{3, 4, 5}}  (1);
	\draw[->] (3) to node[right, midway, align=left] {\footnotesize $1$, \textcolor{blue}{1}} (14);
	
	%outgoing edges v_11
	\draw[->] (11) to[in=43, out=280, looseness=1.1] node[right, right] {\footnotesize 1, \textcolor{blue}{2}} (1);
	\draw[->] (11) to  node[above, midway] {\footnotesize $1$, \textcolor{blue}{1}} (20);
	
	%outgoing edges v_12
	\draw[->] (12) to  node[above, midway, sloped] {\footnotesize 2, \textcolor{blue}{$\delta, \delta+2$}} (11);
	\draw[->] (12) to[bend right=50] node[above, midway, sloped] {\footnotesize 1, \textcolor{blue}{$\delta+1$}} (1);
	\draw[->] (12) to  node[above, pos=0.7, align=right] {\footnotesize $\delta-1$, \textcolor{blue}{1} ,\\ \footnotesize \textcolor{blue}{$\dots$}, \\ \footnotesize \textcolor{blue}{$\delta-1$}} (21);
\end{tikzpicture}
    \end{subfigure}
    \caption{The left image visualizes the structure of a min-clearing game without Nash equilibrium. The game consists of three completely symmetric parts. In the right image, the edge weights (black edge labels), seniorities (blue edge labels) and external assets are given for one of the symmetric parts.}
    \label{fig:min-clearing-sen-NE}
    \end{figure}
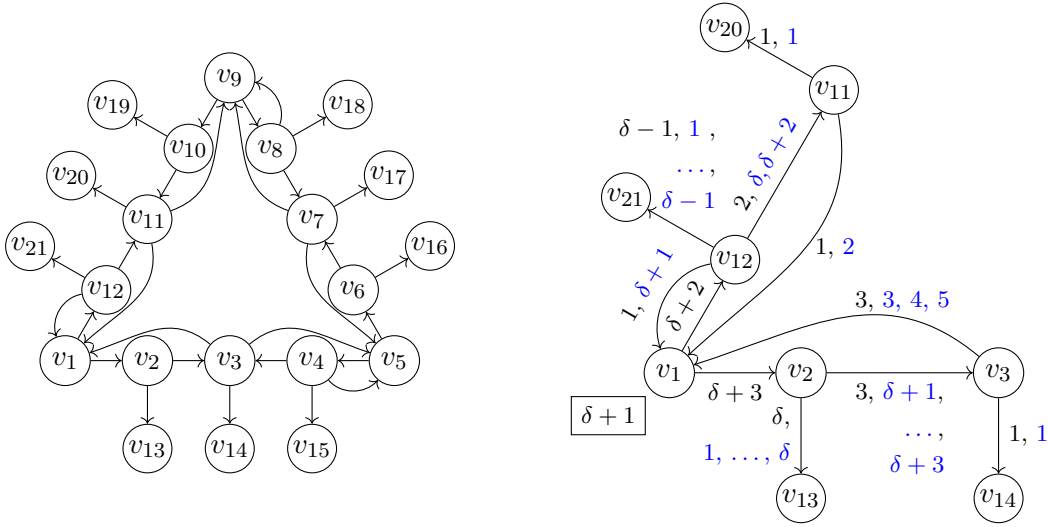
    \noindent
    Consider the game visualized in Figure \ref{fig:min-clearing-sen-NE}. We interpret the depicted graph as a multigraph where all edges have unit-weights. For space reasons, $d$ edges between two players are visualized by a single edge with capacity $d$. Note that the $d$ individual edges may be assigned different seniorities. In the considered game, only the players $v_1, v_5$ and $v_9$ have a meaningful strategy choice, since the behavior of all other players is essentially determined by the fixed partition. Furthermore, only players $v_1, v_5$ and $v_9$ hold external assets and, consequently, all flow in the network originates from these three players. Now consider the strategy choices of player $v_1$. If she pays $\delta$ to $v_{12}$, all assets are passed to $v_{20}$ and $v_{21}$. Only when $v_1$ forwards at least $\delta+\epsilon$ where $\epsilon>0$ can she, in a sense, unlock cycles. Then, she can increase the flow on her incoming edges utilizing these cycles.
    
    A similar observation can be made for the payments of $v_1$ towards $v_2$. Therefore, we can divide all possible strategy choices into three categories: (1) the player pays at least $\delta+\epsilon$ to $v_{12}$, (2) the player pays at least $\delta+\epsilon$ to $v_2$, or (3) the player splits her assets in some other way. For categories (1) and (2), we further assume that $v_1$ uses all other assets to increase flow on the unlocked cycles. Due to symmetry, these categories can be generalized both for players $v_5$ and $v_9$. 
    %It is easy to verify, that there exists at least one player that profits by deviating. 
     The following utility matrices verify nonexistence of a Nash equilibrium. The entries of the matrices should be interpreted as utility of $v_1, v_5$ and $v_9$.
    %The corresponding utility matrices are given in Appendix~\ref{apx:C}.

\vphantom{}
\begin{center}
    \begin{tabular}{x{0.5cm}|x{3.5cm}|x{3.5cm}|x{3.5cm}}
	\multicolumn{4}{c}{$v_9$ : (1)}\\ \hline \hline 
	\diag{.1em}{.5cm}{\footnotesize $v_1$}{\footnotesize $v_5$}& (1) & (2) & (3)\\ \hline
	(1) & $\delta+3,\delta+3,\delta+3$ & $\delta+3,\delta+4,\delta+3$ & $\delta+3,\delta+1,\delta+3$  \\ \hline
	(2) & $\delta+4,\delta+3,\delta+3$ & $\delta+1,\delta+4,\delta+3$ & $\leq\delta+4,\leq\delta+3,\delta+3$ \\ \hline
	(3) & $\delta+1,\delta+3,\delta+3$ & $\delta+1,\delta+4,\delta+3$ & $\delta+1,\delta+1,\delta+3$
\end{tabular}
\vspace{0.2cm}\\
\begin{tabular}{x{0.5cm}|x{3.5cm}|x{3.5cm}|x{3.5cm}}
	\multicolumn{4}{c}{$v_9$ : (2)}\\ \hline \hline 
	\diag{.1em}{.5cm}{\footnotesize $v_1$}{\footnotesize $v_5$}& (1) & (2) & (3)\\ \hline
	(1) & $\delta+3,\delta+3,\delta+4$ & $\delta+3,\delta+1,\delta+4$ & $\delta+3,\delta+1,\delta+4$  \\ \hline
	(2) & $\delta+4,\delta+3,\delta+1$ & $\delta+1,\delta+1,\delta+1$ & $\leq\delta+4,\leq\delta+3,\delta+1$ \\ \hline
	(3) & $\leq\delta+3,\delta+3,\leq\delta+4$ & $\leq\delta+3,\delta+1,\delta+4$ & $\leq\delta+3,\delta+1,\leq\delta+4$
\end{tabular}
\vspace{0.2cm}\\
\begin{tabular}{x{0.5cm}|x{3.5cm}|x{3.5cm}|x{3.5cm}}
	\multicolumn{4}{c}{$v_9$ : (3)}\\ \hline \hline 
	\diag{.1em}{.5cm}{\footnotesize $v_1$}{\footnotesize $v_5$}& (1) & (2) & (3)\\ \hline
	(1) & $\delta+3,\delta+3,\delta+1$ & $\delta+3,\leq\delta+4,\leq\delta+3$ & $\delta+3,\delta+1,\delta+1$  \\ \hline
	(2) & $\delta+4,\delta+3,\delta+1$ & $\delta+1,\leq\delta+4,\leq\delta+3$ & $\leq\delta+4,\leq\delta+3,\delta+1$ \\ \hline
	(3) & $\leq\delta+1,\delta+3,\leq\delta+1$ & $\delta+1,\leq\delta+4,\leq\delta+3$ & $\delta+1,\delta+1,\delta+1$
\end{tabular}
    \end{center}
    
\end{proof}

\begin{theorem}\label{thm:min-clearing-sen-compute-NE}
    For a given min-clearing game with fixed seniorities and \TM\ strategies with any $k\geq 2$, it is \classNP-hard to decide whether a pure Nash equilibrium exists.
\end{theorem}

\begin{proof}
    Consider the game without Nash equilibrium in Figure \ref{fig:min-clearing-sen-NE}. In this game, there exists an equilibrium if the assets of $v_1$ are increased by at least $k-1+\epsilon_f$, for $\epsilon_f>0$. Then, she can pay of all her liabilities when $v_5$ plays according to (2) and $v_9$ according to (1). No player has an incentive to deviate unilaterally to another strategy profile.
    Towards \classNP-hardness, we proceed similarly to proof of Theorem \ref{thm:max-clearing-sen-compute-NE}. Here, the game without equilibrium has been extended to include a construction that directs exactly assets of 1 into the game if there exists a satisfying assignment for an instance from \tSAT. We now generate additional assets of $k$ for player $v_1$ by adding k copies of the construction, each with edge $(u_B, v_1)$. 
\end{proof}

%
% ---- Bibliography ----
%
\bibliographystyle{abbrv}
\bibliography{literature}

\begin{thebibliography}{10}

\bibitem{BertschingerHS20}
N.~Bertschinger, M.~Hoefer, and D.~Schmand.
\newblock Strategic payments in financial networks.
\newblock In {\em Proc.\ 11th Symp.\ Innov.\ Theoret.\ Comput.\ Sci.\ (ITCS)},
  pages 46:1--46:16, 2020.

\bibitem{CsokaH18}
P.~Csóka and P.~Jean-Jacques~Herings.
\newblock Decentralized clearing in financial networks.
\newblock {\em Manag.\ Sci.}, 64(10):4681--4699, 2018.

\bibitem{Diamond93}
D.~Diamond.
\newblock Seniority and maturity of debt contracts.
\newblock {\em J.\ Financial Econ.}, 33(3):341--368, 1993.

\bibitem{EisenbergN01}
L.~Eisenberg and T.~Noe.
\newblock Systemic risk in financial systems.
\newblock {\em Manag.\ Sci.}, 47(2):236--249, 2001.

\bibitem{Elsinger09}
H.~Elsinger.
\newblock {Financial Networks, Cross Holdings, and Limited Liability}.
\newblock Working Paper 156, Oesterreichische Nationalbank (Austrian Central
  Bank), 2009.

\bibitem{ElsingerLS06}
H.~Elsinger, A.~Lehar, and M.~Summer.
\newblock Risk assessment for banking systems.
\newblock {\em Manag.\ Sci.}, 52(9):1301--1314, 2006.

\bibitem{Fischer14}
T.~Fischer.
\newblock No-arbitrage pricing under systemic risk: {A}ccounting for
  cross-ownership.
\newblock {\em Math.\ Finance}, 24(1):97--124, 2014.

\bibitem{Kanellopoulos21}
P.~Kanellopoulos, M.~Kyropoulou, and H.~Zhou.
\newblock Financial network games.
\newblock In {\em Proc.\ 2nd ACM Int.\ Conf.\ AI in Finance (ICAIF)}, 2021.
\newblock To appear.

\bibitem{KrieterR21}
F.~Krieter and D.~Rau.
\newblock Limited liabilities within a (re-)insurance group.
\newblock Working Paper 5694, 2021.

\bibitem{Kusnetsov2018}
M.~Kusnetsov.
\newblock {\em Clearing models for systemic risk assessment in interbank
  networks}.
\newblock PhD thesis, London School of Economics and Political Science, 2018.

\bibitem{PappW20}
P.~A. Papp and R.~Wattenhofer.
\newblock Network-aware strategies in financial systems.
\newblock In {\em Proc.\ 47th Int.\ Colloq.\ Autom.\ Lang.\ Programming
  (ICALP)}, pages 91:1--91:17, 2020.

\bibitem{PappW21ec}
P.~A. Papp and R.~Wattenhofer.
\newblock Debt swapping for risk mitigation in financial networks.
\newblock In {\em Proc.\ 22nd Conf.\ Econ.\ Comput.\ (EC)}, pages 765--784,
  2021.

\bibitem{PappW21wine}
P.~A. Papp and R.~Wattenhofer.
\newblock Default ambiguity: {F}inding the best solution to the clearing
  problem.
\newblock In {\em Proc.\ 17th Conf.\ Web and Internet Econ.\ (WINE)}, pages
  391--409, 2021.

\bibitem{PappW21}
P.~A. Papp and R.~Wattenhofer.
\newblock Sequential defaulting in financial networks.
\newblock In {\em Proc.\ 12th Symp.\ Innov.\ Theoret.\ Comput.\ Sci.\ (ITCS)},
  pages 52:1--52:20, 2021.

\bibitem{RogersV13}
L.~Rogers and L.~Veraart.
\newblock Failure and rescue in an interbank network.
\newblock {\em Manag.\ Sci.}, 59(4):882--898, 2013.

\bibitem{SchuldenzuckerS20}
S.~Schuldenzucker and S.~Seuken.
\newblock Portfolio compression in financial networks: Incentives and systemic
  risk.
\newblock In {\em Proc.\ 21st Conf.\ Econ.\ Comput.\ (EC)}, page~79, 2020.

\bibitem{SchuldenzuckerSB17}
S.~Schuldenzucker, S.~Seuken, and S.~Battiston.
\newblock Finding clearing payments in financial networks with credit default
  swaps is {PPAD}-complete.
\newblock In {\em Proc.\ 8th Symp.\ Innov.\ Theoret.\ Comput.\ Sci.\ (ITCS)},
  pages 32:1--32:20, 2017.

\bibitem{SchuldenzuckerSB20}
S.~Schuldenzucker, S.~Seuken, and S.~Battiston.
\newblock Default ambiguity: Credit default swaps create new systemic risks in
  financial networks.
\newblock {\em Manag.\ Sci.}, 66(5):1981--1998, 2020.

\bibitem{Veraart20}
L.~Veraart.
\newblock When does portfolio compression reduce systemic risk?
\newblock 2020.
\newblock SSRN 10.2139/ssrn.3688495.

\end{thebibliography}

\end{document}